\documentclass[11pt]{amsart}
\usepackage{extsizes}
\usepackage[applemac]{inputenc}
\usepackage[titletoc]{appendix}
\usepackage{pxfonts}

\usepackage[backend=biber, style=numeric, date=year, url=false, doi=false, isbn=false, eprint=true, firstinits=true, maxcitenames=5, maxbibnames=5, style=alphabetic]{biblatex}
\AtEveryBibitem{%
  \ifentrytype{online}{%
  }{%
    \clearfield{eprint}%
    \clearfield{urldate}%
  }%
}
\makeatletter
\def\blx@maxline{77}
\makeatother

\usepackage{blindtext}
\usepackage{amssymb}
\usepackage{amsfonts}
\usepackage{amsthm}
\usepackage{tikz}
\usepackage{caption}
\usepackage{color}
\usepackage{graphicx}
\usepackage{xcolor}
\usepackage{shadethm}
\usepackage{subfigure}
\usepackage[mathcal]{euscript}
\usepackage{epigraph}
\usepackage[T1]{fontenc}

\usepackage{tgpagella}

\usepackage{enumerate}
\usepackage{verbatim}
\usetikzlibrary{trees}
\usetikzlibrary{decorations.markings}
\usetikzlibrary{arrows}
\usepackage{mathrsfs}
\usepackage{xparse}
\usepackage{tikz-cd}

\usetikzlibrary{positioning,arrows,patterns}
\usetikzlibrary{decorations.markings}
\usetikzlibrary{calc}
\tikzset{
  photon/.style={decorate, decoration={snake}, draw=black},
  fermion/.style={draw=black, postaction={decorate},decoration={markings,mark=at position .55 with {\arrow{>}}}},
  vertex/.style={draw,shape=circle,fill=black,minimum size=5pt,inner sep=0pt},
particle/.style={thick,draw=black},
particle2/.style={thick,draw=blue},
avector/.style={thick,draw=black, postaction={decorate},
    decoration={markings,mark=at position 1 with {\arrow[black]{triangle 45}}}},
gluon/.style={decorate, draw=black,
    decoration={coil,aspect=0}}
 }
\NewDocumentCommand\semiloop{O{black}mmmO{}O{above}}
{%
\draw[#1] let \p1 = ($(#3)-(#2)$) in (#3) arc (#4:({#4+180}):({0.5*veclen(\x1,\y1)})node[midway, #6] {#5};)
}

\usepackage{url}
\usepackage{hyperref}
\hypersetup{colorlinks=true, citecolor=purple, linktocpage, linkcolor=blue, urlcolor=cyan} 

\newtheoremstyle{own}
  {3pt}
  {3pt}
  {\sffamily}
  {0pt}
  {\bfseries}
  {.}
  {5pt plus 1pt minus 1pt}
  {}

\numberwithin{equation}{section}

\theoremstyle{plain}

\theoremstyle{definition}
\newtheorem{defn}{Definition}[subsection]
\newtheorem{prop}{Proposition}[subsection]
\newtheorem{lem}{Lemma}[subsection]
\newtheorem{cor}{Corollary}[subsection]

\theoremstyle{remark}
\newtheorem{rem}{Remark}[subsection]

\newtheorem{exe}{Exercise}[subsection]
\newtheorem{ex}{Example}[subsection]

\newcommand{\R}{\mathbb{R}}

\newcommand{\N}{\mathbb{N}}

\newcommand{\dd}{{\mathrm{d}}}

\newcommand{\id}{\mathrm{id}}

\DeclareMathOperator{\End}{End}

\newcommand{\de}{\partial}

\newcommand{\calA}{\mathcal{A}}
\newcommand{\calB}{\mathcal{B}}
\newcommand{\calH}{\mathcal{H}}
\newcommand{\calS}{\mathcal{S}}

\def\gpd{\,\lower1pt\hbox{$\longrightarrow$}\hskip-.24in\raise2pt
               \hbox{$\longrightarrow$}\,}

\newcommand{\I}{\mathrm{i}}

\newcommand{\ee}{\textnormal{e}}

\makeatletter

\pgfdeclareshape{genus}{
 \anchor{center}{\pgfpointorigin}
\backgroundpath{
    \begingroup
    \tikz@mode
    \iftikz@mode@fill

         \pgfpathmoveto{\pgfqpoint{-0.78cm}{-.17cm}}
     \pgfpathcurveto %
            {\pgfpoint{-0.35cm}{-.44cm}}
        {\pgfpoint{0.35cm}{-.44cm}}
        {\pgfpoint{.78cm}{-0.17cm}} 
     \pgfpathmoveto{\pgfqpoint{-0.78cm}{-0.17cm}}
     \pgfpathcurveto %
            {\pgfpoint{-0.25cm}{.25cm}}
        {\pgfpoint{.25cm}{.25cm}}
        {\pgfpoint{0.78cm}{-0.17cm}}
        \pgfusepath{fill}
        \fi

        \iftikz@mode@draw
     \pgfpathmoveto{\pgfqpoint{-1cm}{0cm}}
     \pgfpathcurveto %
            {\pgfpoint{-0.5cm}{-.5cm}}
        {\pgfpoint{0.5cm}{-.5cm}}
        {\pgfpoint{1cm}{0cm}}

     \pgfpathmoveto{\pgfqpoint{-0.75cm}{-0.15cm}}
     \pgfpathcurveto %
            {\pgfpoint{-0.25cm}{.25cm}}
        {\pgfpoint{.25cm}{.25cm}}
        {\pgfpoint{0.75cm}{-0.15cm}}
              \pgfusepath{stroke}
        \fi
        \endgroup
    }
    }

    \makeatother

\begin{document}

\title[Notes on Geometric Quantization]{
Notes on Geometric Quantization}
\author[N. Moshayedi]{Nima Moshayedi}
\address{Institut f\"ur Mathematik\\ Universit\"at Z\"urich\\ 
Winterthurerstrasse 190
CH-8057 Z\"urich}
\email[N.~Moshayedi]{nima.moshayedi@math.uzh.ch}

\maketitle

%
%

\begin{abstract}
These notes give an introduction to the quantization procedure called \emph{geometric quantization}. It gives a definition of the mathematical background for its understanding and introductions to classical and Quantum Mechanics, to differentiable manifolds, symplectic manifolds and the geometry of line bundles and connections. Moreover, these notes are endowed with several exercises and examples.
\end{abstract}

\tableofcontents

\section{Motivation}
Quantization procedures are of strong mathematical interest and there are different approaches to quantization. \emph{Geometric quantization} \cite{GuilleminSternberg1982,Kir85,Wood97,BatesWeinstein2012} lies the focus on constructing the mathematical structure of a Hilbert space by constructing a certain line bundle with a particular type of connection, whereas e.g. \emph{deformation quantization} \cite{Weyl1931,Moy,DeWildeLecomte1983,Fedosov1994,K,CalaquePantevToenVaquieVezzosi2017} uses the noncommutativity structure of quantum observables, where it deforms the classical product on the Poisson algebra to a star product (see e.g. \cite{GuttRawnsleySternheimer2005}). However, there the Hilbert space of states is not constructed explicitly, which is of importance to understand perturbative quantum field theories also in relation to the Atyiah--Segal formulation \cite{Atiyah1988,Segal1988} of topological quantum field theories. 
Moreover, geometric quantization uses the symplectic structure of the classical setting, whereas deformation quantization uses the Poisson structure.


\section{Introduction to Classical Mechanics}
\subsection{Newton's law of motion} In Classical Mechanics, we have Newton's axioms:
\begin{enumerate}
\item{Every particle remains at rest or moves with a constant speed, unless acted upon by a force.}
\item{Rate of change of momentum $=$ Force.}
\item{To every action on a particle, there is an equal and opposition reaction.}
\end{enumerate}

We want to focus on Newton's second law of motion. 
\begin{defn}[Particle]
A \emph{particle} is an object of insignificant size, i.e. the only information we have about a particle is its position for a given time.
\end{defn}
\begin{ex}
Examples of particles are: electrons, tennis balls, cars, planets, etc.
\end{ex}

To describe the position of a particle, we need a reference frame (coordinate system). We denote by $x=x(t)$ the curve, which is the trajectory of a given particle. Moreover, we define the velocity of a given trajectory $x(t)$ by $v(t)=\frac{\dd x(t)}{\dd t}=\dot{x}(t)$, and the acceleration by $a(t)=\frac{\dd v}{\dd t}=\frac{\dd^2x}{\dd t^2}=\ddot{x}(t)$. We write $p:=mv$ for the momentum. Then we can write Newton's second law of motion as 
\begin{equation}
\label{second_law}
F=ma,
\end{equation}
where we assume $m$ to be constant (mass). This is a second order ordinary differential equation. 

\begin{ex}[Free particle]
Consider the free particle, i.e. $F=0$. Then, $m\ddot{x}(t)=0$ and thus $x(t)=x_0+v t$, where $v$ is the initial velocity and $x_0$ is the initial position.
\end{ex}

\begin{rem}
We would like to know whether Newton's second law of motion implies the first law of motion. The answer is: it is complicated. We already assume Newton's first law of motion for the second law. 
\end{rem}

\begin{ex}[Particle in the presence of a conservative force]
\label{part_cons_force}
Let us start first with the following definition:

\begin{defn}[Conservative force]
We say a force $F$ is \emph{conservative}, if $F=F(x)$, i.e. it only depends on the position.
\end{defn}

For a conservative force, we can define the potential $V(x)$ by the equation
\begin{equation}
\label{cons_force}
F(x)=-\nabla V,
\end{equation}
where $\nabla=(\partial_{x_1},...,\partial_{x_n})$, and everything is sufficiently \emph{nice}, such that 
\[
V(x)=-\int_{x_0}^xF(u)\dd u.
\]
We will work with Equation \eqref{cons_force} as a definition of conservative force. 
\end{ex}
With Example \ref{part_cons_force} we see that Equation \eqref{second_law} becomes 
\begin{equation}
\label{pot_second_law}
-\nabla V=m\ddot{x}.
\end{equation}
\begin{rem}
$V$ is also called the \emph{potential energy}.
\end{rem}

We want to justify the word conservative for such a force. In this situation, there is a conserved quantity, called the total energy $E=\frac{1}{2}m\dot{x}^2+V(x)$, where the first term is called the \emph{kinetic} energy. 

\begin{lem}
If $x(t)$ satisfies Newton's equation of motion, then $\frac{\dd E}{\dd t}=0$ along $x(t)$, i.e. $E$ is conserved.
\end{lem}

\begin{proof}
We have $E=\frac{1}{2}m\dot{x}^2+V(x)$. Then 
\[
\frac{\dd E}{\dd t}=m\ddot{x}\dot{x}+\frac{\dd V}{\dd t}\frac{\dd x}{\dd t}=\dot{x}\left(m\ddot{x}+\frac{\dd V}{\dd t}\right)=0.
\]
\end{proof}
Now, since $E$ is constant, we can write $\frac{1}{2}m\dot{x}^2=E-V(x)$ and hence $\dot{x}=\frac{2}{m}\sqrt{E-V(x)}$, i.e. 
\begin{equation}
\label{id1}
\frac{\dd x}{\sqrt{E-V(x)}}=\dd t.
\end{equation}

Solving Newton's equation of motion (which is a second order ODE) can be reduced to solving a first order ODE \eqref{id1}. This shows that the existence of conserved quantities can be useful for solving equations of motion.

\begin{ex}[Harmonic oscillator]
The \emph{harmonic oscillator} is described by the potential $V(x)=\frac{1}{2}kx^2$, where $k$ is some constant. Thus, $F(x)=-kx$ (Hooke's law), and the equation of motion (without friction) is given by 
\[
m\ddot{x}+kx=0.
\]
It has a general solution of the form $x(t)=A\cos(\omega t)+B\sin(\omega t)$, with $\omega=\sqrt{\frac{k}{m}}$ and $A,B$ some constants.
\end{ex}

\begin{ex}[Uniform gravitational field]
 Let $g$ denote the acceleration due to gravity. Consider the potential $V(z)=mgz$ (thus $F=-mg$). Then we get the equation of motion 
 \[
 \ddot{z}=-g.
\]
Solving this, we get $z(t)=z_0+vt-\frac{1}{2}gt^2$. 
\end{ex}

\subsection{Newton's principle of determinism}
The initial state of a mechanical system (the totality of positions and velocities of its points at some moment) uniquely determines all of its motion. E.g. for a particle moving on a line the possible states are given by the set $\{(a,b)\mid a,b\in\R\}$. The modern point of view would be to regard Newton's equation as a second order ODE, hence it is enough to specify two initial conditions to solve the equation of motion.

\subsection{Hamiltonian mechanics}
Consider a Newtonian mechanical system, where a particle is moving in $\R^n$ in the presence of a conservative force $F=-\nabla V$. Recall that $E(x,v)=\frac{1}{2}mv^2+V(x)$ with equation of motion $-\frac{\partial V}{\partial x^j}=m\ddot{x}^j$, for $j=1,...,n$.

\begin{defn}[Momentum]
We call $p:=mv$ the \emph{mechanical (linear) momentum} of the system.
\end{defn}

We can write $E(x,p)=\frac{p^2}{2m}+V(x)$ and thus $\dot{x}_j=\frac{\partial E}{\partial p_j}$. Moreover, $\dot{p}_j=m\ddot{x}^j=-\frac{\partial V}{\partial x^j}=-\frac{\partial E}{\partial x_j}$. Hence, we get a system of first order ODEs
\begin{align}
\dot{x}^j&=\frac{\partial E}{\partial p_j},\\
\dot{p}_j&=-\frac{\partial E}{\partial x^j},
\end{align}
for $j=1,...,n$. 
\begin{rem}
We have not achieved anything new except for rewriting Newton's equation as a system of first order equation.
\end{rem}

\begin{defn}[Phase space]
The space $\R^{2n}\ni (x,p)$ is called the \emph{phase space} (or simply \emph{state space} for the mechanical system). Here $x=(x^1,...,x^n)$ and $p=(p_1,...,p_n)$.
\end{defn}

\begin{defn}[Hamilton's equations]
Given a function $H\in C^\infty(\R^{2n})$, we can consider the system of equations 
\begin{align}
\dot{x}^j&=\frac{\partial H}{\partial p_j},\\
\dot{p}_j&=-\frac{\partial H}{\partial x^j},
\end{align}
for $j=1,...,n$, called \emph{Hamilton's equations}.
\end{defn}

\begin{rem}
In Newtonian mechanics, we studied equations of motion in the configuration space (space of all possible positions), where as in the Hamiltonian approach, we will consider the phase space. We would like $\R^{2n}$ to may have other structures, which can be useful to prove the equations of motions.
\end{rem}

\subsection{Poisson bracket}

Given smooth functions $f$ and $g$ on $\R^{2n}$, we can define 
\begin{equation}
\label{Poisson_bracket}
\{f,g\}:=\sum_{j=1}^n\left(\frac{\partial f}{\partial x^j}\frac{\partial g}{\partial p_j}-\frac{\partial g}{\partial x^j}\frac{\partial f}{\partial p_j}\right).
\end{equation}
The map $\{\enspace,\enspace\}\colon C^\infty(\R^{2n})\times C^\infty(\R^{2n})\to C^\infty(\R^{2n})$ is called the \emph{Poisson bracket}.  

\begin{exe}
Show that the Poisson bracket satisfies for all $f,g,h\in C^\infty(\R^{2n})$ the following properties:
\begin{enumerate}[$(i)$]
\item{$\{f,g\}=-\{g,f\}$,
}
\item{$\{\enspace,\enspace\}$ is $\mathbb{C}$-bilinear,
}
\item{$\{f,gh\}=\{f,g\}h+\{f,h\}g$, i.e. $\{f,\enspace\}$ is a derivation,
}
\item{$\{f,\{g,h\}\}=\{\{f,g\},h\}+\{g,\{f,h\}\}$ (Jacobi identity).
}
\end{enumerate}
\end{exe}

\begin{ex}
Let $p_j\colon \R^{2n}\to \R$, $p_j(x,p)=p_j$, and $x^j\colon \R^{2n}\to \R$, $x^j(x,p)=x^j$. Then 
\begin{itemize}
\item{$\{x^{i},x^j\}=\{p_i,p_j\}=0$,
}
\item{$\{x^{i},p_j\}=\delta_{ij}$.}
\end{itemize}
\end{ex}

\begin{prop}
Let $f\in C^\infty(\R^{2n})$, then $$\frac{\dd f}{\dd t}=\{f,H\}$$ along a solution $\{(x(t),p(t))\}$ of Hamilton's equations.
\end{prop}

\begin{proof} We have
$$\frac{\dd f}{\dd t}=\sum_{j=1}^n\left(\frac{\partial f}{\partial x^j}\frac{\dd x^j}{\dd t}+\frac{\partial f}{\partial p_j}\frac{\dd p_j}{\dd t}\right)=\sum_{j=1}^n\left(\frac{\partial f}{\partial x^j}\frac{\partial H}{\partial p_j}+\frac{\partial f}{\partial p_j}\frac{\partial H}{\partial x^j}\right)=\{f,H\}.$$
\end{proof}

\begin{cor}
Let $f\in C^\infty(\R^{2n})$. Then $f$ is conserved along a solution $(x(t),p(t))$ of Hamilton's equations, and hence $\{f,H\}=0$ along the solutions.
\end{cor}

\begin{ex}
$\{H,H\}=0$ implies $H$ is conserved.
\end{ex}
\begin{ex}
$f,g$ conserved implies $\{f,g\}$ is conserved.
\end{ex}

\section{Symplectic linear algebra}

\subsection{Symplectic vector spaces} Let $V$ be a finite-dimensional vector space over $k=\R$ or $\mathbb{C}$. Denote by $V^*$ the dual of $V$. An element of $V^*$ is a $k$-linear map $f\colon V\to k$. Let $0\leq m\leq \dim V$. Define 
\begin{multline*}
\bigwedge^mV^*:=\Big\{\phi\colon \overbrace{V\times\dotsm \times V}^{m}\to k\,\,\big|\,\, \text{ $\phi$ is linear in each argument and $\phi$ is alternating,}\\
\text{ i.e. $\phi(v_1,...,v_j,v_{j+1},...,v_m)=-\phi(v_1,...,v_{j+1},v_j,...,v_m)$ for all $j=1,...,m-1$}\Big\}
\end{multline*}

\begin{ex}
Let $f,g\in V^*$. Then we can define $(f\land g)\in \bigwedge^2 V^*$ by 
\[
(f\land g)(v_1,v_2)=f(v_1)g(v_2)-f(v_2)g(v_1).
\]
In fact, it can be shown that all the elements of $\bigwedge^2V^*$ are finite linear combinations of such elements. Given $\Omega\in \bigwedge^2 V^*$, we can define a map
\begin{align*}
\Omega^\flat\colon V&\to V^*\\
v&\mapsto \Omega^\flat(v),
\end{align*}
where $\Omega^\flat(v)(w):=\Omega(v,w)$.
\end{ex}

\begin{defn}[Symplectic vector space]
A \emph{symplectic vector space} is a pair $(V,\Omega)$, where $V$ is a (finite-dimensional) vector space and $\Omega\in\bigwedge^2 V^*$ such that $\Omega^\flat$ is a vector space isomorphism.
\end{defn}

\begin{rem}
Since we are in the finite-dimensional setting, $\Omega^\flat$ is a vector space isomorphism if and only $\Omega^\flat$ is injective.
\end{rem}

\begin{rem}
\label{rem1}
$\Omega^\flat$ is in fact injective if and only if there is a $v\in V$ such that $\Omega(v,w)=0$ for all $w\in V$ implies $v=0$.
\end{rem}

\begin{ex}
Let $(W,\langle\enspace,\enspace\rangle)$ be an inner product space. Consider $V:=W\oplus W$, with $\Omega((w_1,w_2),(w_1',w_2')):=\langle w_2',w_1\rangle-\langle w_2,w_1'\rangle$. Then $(V,\Omega)$ is a real symplectic vector space. More generally, if $V:=W\oplus W^*$ and $\Omega_{can}((w,\alpha),(w',\alpha')):=\alpha'(w)-\alpha(w')$, then $(V,\Omega_{can})$ is a symplectic vector space.
\end{ex}

\begin{rem}
Note that the fact that $(V,\Omega_{can})$ is a symplectic vector space is implied by Remark \ref{rem1}.
\end{rem}

\begin{defn}[Isotropic/Coisotropic/Lagrangian]
Let $(V,\Omega)$ be a symplectic vector space. Let $Y$ be a subspace of $V$. Define the symplectic complement of $Y$ by $Y^{\perp}:=\{v\in V\mid \Omega(v,y)=0,\forall y\in Y\}$. Then 
\begin{itemize}
\item{$Y$ is \emph{isotropic} if $Y\subseteq Y^\perp$,
}
\item{$Y$ is \emph{coisotropic} if $Y^\perp\subseteq Y$,}
\item{$Y$ is \emph{Lagrangian} if $Y$ is isotropic and $Y$ is symplectic if $\Omega\big|_{Y\times Y}$ is nondegenerate, i.e. $Y\cap Y^\perp=\{0\}$.}
\end{itemize}
\end{defn}

\begin{ex}
If $\dim Y=1$, then $Y$ is isotropic. If $Y$ is isotropic, then $Y^\perp$ is coisotropic. If $Y$ is symplectic, then so is $Y^\perp$. Moreover, $Y^{\perp\perp}=(Y^\perp)^\perp=Y$.
\end{ex}

\begin{prop}
Let $(V,\Omega)$ be a symplectic vector space. Then there is a basis $\{e_1,...,e_n,f_1,...,f_n\}$ of $V$ such that 
\begin{align*}
\Omega(e_i,e_j)&=0,\\
\Omega(f_i,f_j)&=0,\\
\Omega(e_i,f_j)&=\delta_{ij}.
\end{align*}
for all $i,j\in\{1,...,n\}$. Hence, we can write $\Omega=\sum_{j=1}^ne_j^*\land f_j^*$.
\end{prop}

\begin{rem}
Note that $\dim V=\dim Y+\dim Y^\perp$. Moreover, $Y$ is symplectic if and only if $V=Y\oplus Y^\perp$.
\end{rem}

\begin{rem}
It is easy to see that $Y$ is a Lagrangian subspace if and only if $Y$ is isotropic and $\dim Y=\frac{1}{2}\dim V$. Moreover, $Y$ is Lagrangian if and only if $Y$ is a maximal isotropic subspace.
\end{rem}

\subsection{K\"ahler structure}

Let $V$ be a real vector space.

\begin{defn}[Complex structure]
A \emph{complex structure} $J$ on $V$ is a linear map $J\colon V\to V$ such that $J^2=-\id$. 
\end{defn}
\begin{rem}
If $V$ has a complex structure $J$, then $V$ can be turned into a complex vector space $V_J$ by $(a+\I b)v=av+bJv$. In particular, $\dim_\R V$ is even.
\end{rem}

\begin{rem}
In fact, if $\dim V$ is even, one can show that $V$ carries a complex structure.
\end{rem}

Assume $V$ is a real vector space and $\dim V=2n$. Let $V^\mathbb{C}$ denote the \emph{complexification} of $V$, i.e. $V^\mathbb{C}=V\oplus \I V$, or equivalently $V^\mathbb{C}=V\otimes_\R \mathbb{C}$. We get $(a+\I b)(v,w)=(av-bw,bv+aw)$. Then $V^\mathbb{C}$ is a complex vector space. If $\{e_1,...,e_n\}$ is an $\R$-basis of $V$, then $\{(e,0),...,(e_n,0)\}$ is a $\mathbb{C}$-basis of $V^\mathbb{C}$. Let $J$ be a complex structure on $V$ and let $J$ denote the complex linear extension of $V$ to $V^\mathbb{C}$, i.e. $J(a+\I b)v=(a+\I b)Jv$. Moreover, denote by $F_J$ the $+\I$-eigenspace of $J$.

\begin{defn}[Complex conjugation]
Let $c\colon V^\mathbb{C}\to V^\mathbb{C}$ be the complex antilinear map, $c(\alpha\otimes v):=\bar\alpha\otimes v$. We call $c$ the \emph{complex conjugation}. given a subspace $W\subseteq V^\mathbb{C}$, we will write $\overline{W}$ for $c(W)$. 
\end{defn}

\begin{rem}
In particular, we have $V^\mathbb{C}=F_J\oplus \bar F_J$. Moreover, consider the map $V_J\to F_J$, $v\mapsto v-\I Jv$. Then one can show that this map is an isomorphism of complex vector spaces.
\end{rem}

Let $F$ be a subspace of $V^\mathbb{C}$ such that $\dim F=n$ and $V^\mathbb{C}=F\oplus \bar F$. Define a complex linear map $J_F\colon V^\mathbb{C}\to V^\mathbb{C}$ by declaring that $F$ is the $+\I$-eigenspace of $J_F$ and $\bar F$ is the $-\I$-eigenspace of $J_F$. 

\begin{lem}
$J_F$ induces a complex structure on $V$.
\end{lem}

\begin{proof}
We want to show first that $J_F(V)\subseteq V$. Let $v\in V$. Then, we can write $v=f+\bar f$, where $f\in F$. In particular, if $v=f+g$, then $g=\bar f$. Indeed, since $\bar v=v$, we have $\bar f+\bar g=f+g$ and thus $\bar f=g$. This shows that
$$J_Fv=J_Ff+J_F\bar f=\I f-\I \bar f\Rightarrow \overline{J_Fv}=-\I \bar f+\I f=\I f-\I \bar f=J_Fv,$$
and hence $J_Fv\in V$. Obviously, we have $J_F^2=-\id$.
\end{proof}

We have shown that there is a one-to-one correspondence between the set of all complex structures on $V$ and $\{F\subseteq V^\mathbb{C}$ subspace $\mid$ $\dim F=n$, $V^\mathbb{C}=F\oplus \bar F\}$. 

\begin{rem}
From now on we assume that $(V,\Omega)$ is a real symplectic vector space. We extend $\Omega$ complex bilinearly to $V^\mathbb{C}$, which will be again denoted by $\Omega$.
\end{rem}

\begin{defn}[Symplectomorphism]
A map $T\colon V\to V$ such that $T$ is a vector space isomorphism and $\Omega(Tu,Tv)=\Omega(u,v)$ for all $u,v\in V$ is called a \emph{linear symplectomorphism}.
\end{defn}

\begin{lem}
\label{lem1}
Let $J$ be a complex structure on $V$ such that $\Omega(J_u,J_v)=\Omega(u,v)$ for all $u,v\in V$. Then $F_J$ is a Lagrangian subspace of $V^\mathbb{C}$.
\end{lem}

\begin{lem}[Converse of Lemma \ref{lem1}]
Let $F$ be a Lagrangian subspace of $V^\mathbb{C}$ such that $V=F\oplus \bar F$. Let $J_F$ be the associated complex structure on $V$ corresponding to $F$. Then $\Omega(J_Fu,J_Fv)=\Omega(u,v)$ for all $u,v\in V$.
\end{lem}

\begin{defn}[Compatible/positive structure]
Let $J$ be a complex structure on $(V,\Omega)$. We say that $J$ is \emph{compatible} with $\Omega$ if $\Omega(Ju,Jv)=\Omega(u,v)$ for all $u,v\in V$. We say $J$ is \emph{positive} if $\Omega(u,Ju)>0$ for all $u\in V\setminus\{0\}$.
\end{defn}

\begin{lem}
If $J$ is a compatible positive complex structure on $J$, then the form 
$$(u,v)\mapsto \Omega(u,Jv)$$
defines an inner product on $V$.
\end{lem}

\begin{proof}
We need to check that $\Omega(\enspace,J\enspace)$ is a symmetric bilinear form and it is positive-definite. Bilinearity is clear. It is symmetric since $\Omega(u,Jv)=\Omega(Ju,J^2v)=-\Omega(Ju,v)=\Omega(v,Ju)$. It is positive-definite since $\Omega(u,Ju)>0$ for all $u\in V\setminus\{0\}$, which was an assumption. 
\end{proof}
Let $J$ be a compatible complex structure on $V$. Let $F_J$ be the corresponding Lagrangian subspace of $V^\mathbb{C}$. Define a Hermitian form $h^J$ on $F_J$ by 
$$h^J(u,v):=-\I\Omega(u,\bar v).$$
We actually need to check that it is indeed Hermitian, i.e. it is sesquilinear and $h^J(u,v)=\overline{h^J(v,u)}$.

\begin{lem}
If $J$ is positive then $h^J$ is a positive-definite Hermitian form on $F_J$. 
\end{lem}

\begin{proof}
Given $u\in F_J$, we can write $u=w-\I Jw$, where $w\in V$. Thus
$$\Omega(u,\bar u)=\Omega(w-\I Jw,w+\I Jw)=-\I\Omega(Jw,w)+\I\Omega(w,Jw)=2\I\Omega(w,Jw),$$
and hence $\I\Omega(u,\bar u)=2\Omega(w,Jw)>0$ for $w\in V$.
\end{proof}

\begin{defn}[Hermitian form w.r.t. a Lagrangian subspace]
Let $F$ be a Lagrangian subspace of $V^\mathbb{C}$ such that $V^\mathbb{C}=F\oplus \bar F$. Define a form $h^F$ on $F$ by 
$$h^F(u,v):=-\Omega(u,\bar v).$$
We can easily check that $h^F$ is a non Hermitian form. We say $F$ is \emph{positive} if $h^F$ is positive-definite, i.e. $h^F(u,u)>0$ for $u\in F\setminus\{0\}$.
\end{defn}

\begin{prop}
Let $(V,\Omega)$ be a real symplectic vector space. Then there is a canonical bijection between the set of compatible positive complex structures on $V$ and positive Lagrangian subspaces $F$ of $V^\mathbb{C}$ such that $V^\mathbb{C}=F\oplus\bar F$.
\end{prop}

\begin{defn}[K\"ahler triple]
A triple $(V,\Omega,J)$, where $(V,\Omega)$ is a real symplectic vector space and $J$ is a compatible positive complex structure on $(V,\Omega)$ is called a \emph{K\"ahler triple}.
\end{defn}

\subsection{Exercises}

\begin{exe}
\label{problem1}
Consider a particle moving on the real line in the presence of a force from a potential $V$. Let $E_0\in\R$ be the energy of the particle and suppose $V(x)<E_0$ for all $x_0\leq x\leq x_1$. Then a particle with initial position $x_0$ and a positive initial velocity will continue to move to the right of $x_0$ until it reaches $x_1$. Show that the total time needed to travel from $x_0$ to $x_1$ is given by $$t=\int_{x_0}^{x_1}\sqrt{\frac{m}{2(E_0-V(y))}}\dd y.$$
\end{exe}

\begin{exe}
We will use the notation of Exercise \ref{problem1}. Assume that $V(x)<E_0$ for $x_0\leq x\leq x_1$ but $V(x_1)=E_0$.
\begin{enumerate}
\item{Show that if $V'(x_1)\not=0$, then the particle reaches to $x_1$ in a finite time.}
\item{If $V'(x_1)=0$, then the particle never reaches $x_1$, i.e. the integral $\int_{x_0}^{x_1}\sqrt{\frac{m}{2(E_0-V(y))}}\dd y$ diverges.
}
\end{enumerate}
\end{exe}

\begin{exe}
Let $F$ be a function $F\colon \R^2\setminus\{0\}\to \R^2$ given by $$F(x^1,x^2)=\left(\underbrace{-\frac{x^2}{(x^1)^2+(x^2)^2}}_{F_1},\underbrace{\frac{x^1}{(x^1)^2+(x^2)^2}}_{F_2}\right).$$
Show that: 
\begin{enumerate}
\item{$\frac{\partial F_1}{\partial x^2}=\frac{\partial F_2}{\partial x^1}$,
}
\item{$F$ is not \emph{conservative}, i.e. $F$ is not of the form $-\nabla V$.}
\end{enumerate}
\end{exe}

\begin{exe}
Consider a particle moving in $\R^d$ with a velocity dependent force law $$F(x,v)=-\nabla V(x)+F_2(x,v),$$ where $F_2\colon\R^2\times\R^d\to\R^d$. Assume that $vF_2(x,v)=0$ for all $x,v\in\R^d$. Show that then the function $E(x,v)=\frac{1}{2}mv^2+V(x)$ is conserved.
\end{exe}

\begin{exe}[Angular momentum]
Consider a particle moving in $\R^2$ with position $x$ and velocity $v$. Recall that the momentum is given by $p=mv$. Define the \emph{angular momentum} of the particle by $J=x_1p_2-x_2p_1$. Suppose we have a particle of mass $m$ moving in $\R^2$ under the influence of a conservative force with potential $V(x)$. Show that:
\begin{enumerate}
\item{If $V$ is rotationally invariant in $\R^2$, i.e. $V(x)=V(Ax)$ for any rotation matrix $A=\begin{pmatrix}\cos\theta&-\sin\theta\\ \sin\theta&\cos\theta\end{pmatrix}$, then $J$ is conserved along a solution of Newton's equation.
}
\item{If $J$ is conserved along any solution of Newton's equation, then $V$ is rotationally invariant.
}
\end{enumerate}
\end{exe}

\begin{exe}
Varify the various properties of the Poisson bracket.
\end{exe}

\section{Differentiable manifolds}
Let $M$ be a topological space, which is Hausdorff and second countable.
\begin{defn}[Smooth manifold]
Let $M$ be a topological space and $p\in M$. A chart (or local coordinate system) at $p$ is a pair $(U,\phi)$, where $U\subseteq M$ is an open set containing $p$ and $\phi\colon U\to \R^n$ is a homeomorphism of $U$ onto $\phi(U)$, which is an open subset of $\R^n$. More precisely, such a chart is called a chart of \emph{rank} $n$. Moreover, let $(U,\phi)$ and $(V,\psi)$ be two coordinate charts of rank $n$ such that $U\cap V\not=\varnothing$. Then we have the maps
\begin{align*}
\phi\circ\psi^{-1}\colon \psi(U\cap V)&\to \phi(U\cap V)\\
\psi\circ\phi^{-1}\colon \phi(U\cap V)&\to \psi(U\cap V).
\end{align*}
We say $(U,\phi)$ and $(V,\psi)$ are \emph{smoothly compatible} if $\phi\circ\psi^{-1}$ is smooth as well as $\psi\circ \phi^{-1}$ is smooth, in other words $\phi\circ\psi^{-1}$ is a \emph{diffeomorphism}. Note that $\phi(u\cap V)\subseteq \R^n$ and $\psi(U\cap V)\subseteq\R^n$, and for this situation we know how to define the term \emph{smooth} and \emph{diffeomorphism} etc. A smooth \emph{atlas} $\calA$ of rank $n$ of $M$ is a collection $\{(U_i,\phi_i)\mid i\in I\}$ of smooth compatible coordinate charts such that $M=\bigcup_{i\in I}U_i$. An atlas $\calA$ of rank $n$ is called \emph{maximal} if $\calA$ is not contained in a strictly larger smooth atlas of rank $n$. A \emph{smooth structure} $\calA$ on $M$ is a smooth maximal atlas of rank $n$. The pair $(M,\calA)$ is called a \emph{smooth manifold} of dimension $n$.
\end{defn}

\begin{lem}
\label{atlas}
Given any smooth atlas $\calA$ of $M$, there is a maximal smooth atlas $\overline{\calA}\supseteq \calA$.
\end{lem}

\begin{proof}
See \cite{Lee}.
\end{proof}

\begin{rem}
As a consequence of Lemma \ref{atlas}, we see that it is sufficient to have a smooth atlas in order to define a smooth structure on a topological space $M$.
\end{rem}

\begin{ex}
Let $M=\R^n$. Then $\calA=\{\id_{\R^n}\}$ defines a smooth structure on $\R^n$, which is called the \emph{standard smooth structure}. Note that $\calA$ is not maximal.
\end{ex}

\begin{ex}
Let $M=S^1=\{(x,y)\in\R^2\mid x^2+y^1=1\}$. Define 
\begin{align*}
U_1&=\{(x,y)\in S^1\mid -1<x<1,y>0\},\\
U_2&=\{(x,y)\in S^1\mid -1<y<1,x>0\},\\
U_3&=\{(x,y)\in S^1\mid -1<x<1,y<0\},\\
U_4&=\{(x,y)\in S^1\mid -1<y<1,x<0\}.
\end{align*}
Moreover, define
\begin{align*}
\phi_1\colon U_1&\to (-1,1),\hspace{0.3cm}\phi_1(x,y)=x,\\
\phi_2\colon U_2&\to (-1,1),\hspace{0.3cm}\phi_2(x,y)=y,\\
\phi_3\colon U_3&\to (-1,1),\hspace{0.3cm}\phi_3(x,y)=x,\\
\phi_4\colon U_4&\to (-1,1),\hspace{0.3cm}\phi_4(x,y)=y.\\
\end{align*}
\end{ex}

\begin{exe}
Show that $\{(U_i,\phi_i)\mid i=1,3,4\}$ form an atlas of $S^1$.
\end{exe}

\begin{ex}
Let $V_1=S^1\setminus\{(0,1)\}$ and $V_2=S^1\setminus\{(0,-1)\}$. We call $(0,1)$ the \emph{North pole} and $(0,-1)$ the \emph{South pole}. Define 
\begin{align*}
\psi_1\colon V_1&\to \R, \hspace{0.3cm}\psi_1(x,y)=\frac{x}{1-y},\\
\psi_2\colon V_2&\to \R, \hspace{0.3cm}\psi_2(x,y)=\frac{x}{1+y}.
\end{align*}
\end{ex}

\begin{exe}
Show that $\{(V_1,\psi_1),(V_2,\psi_2)\}$ form another atlas for $S^1$. What can we say about the smooth structures on $S^1$ given by two different atlases $\{(U_i,\phi_i)\mid i=1,3,4\}$ and $\{(V_i,\psi_i)\mid i=1,2\}$?
\end{exe}

\begin{defn}[Smooth functions]
Let $(M,\calA)$ be a smooth manifold of dimension $n$. A function $f\colon M\to k$ (for $k=\R$ or $\mathbb{C}$) is \emph{smooth} if for every coordinate chart $(U,\phi)$, $f\circ \phi^{-1}\colon \phi(U)\to k$ is smooth. We will use $C^\infty(M)$ to denote the set of all \emph{smooth functions} on $M$.
\end{defn}

\begin{defn}[Vector field]
A \emph{vector field} $X$ on $M$ is a map $X\colon C^\infty(M)\to C^{\infty}(M)$ such that for all $f,g\in C^\infty(M)$ and $c\in k$
\begin{enumerate}
\item{$X(f+g)=X(f)+X(g)$ and $X(cf)=cX(f)$,
}
\item{$X(fg)=fX(g)+X(f)g$, i.e. $X$ is a \emph{derivation}.
}
\end{enumerate}
\end{defn}

\begin{ex}
Let $M=\R^n$. Then for $f_1,...,f_n\in\ C^\infty(\R^n)$
$$X=f_1\frac{\partial}{\partial x^1}+\dotsm +f_n\frac{\partial}{\partial x^n}$$
is a vector field.
\end{ex}

\begin{rem}
Let $X$ and $Y$ be two vector fields on $M$. Then $XY$ is \emph{not} a vector field in general. Instead, $[X,Y]:=XY-YX$ is a vector field.
\end{rem}

\begin{defn}[Lie bracket]
For two vector fields $X$ and $Y$ on $M$, we can define their \emph{Lie bracket}
$$[X,Y]:=XY-YX.$$
\end{defn}

Let $\mathrm{Vect}(M)$ denote the space of vector fields on $M$. Then $[\enspace,\enspace]\colon \mathrm{Vect}(M)\times \mathrm{Vect}(M)\to \mathrm{Vect}(M)$ is $k$-bilinear. Moreover, $[X,[Y,Z]]=[[X,Y],Z]+[Y,[X,Z]]$ (Jacobi identity). In other words, $(\mathrm{Vect}(M),[\enspace,\enspace])$ is a \emph{Lie algebra}. 

\begin{defn}[Tangent vector]
Given $p\in M$, a tangent vector $v$ at $p$ is an $\R$-linear map $v\colon C^\infty(M)\to\R$ such that $v(fg)=f(p)v(g)+g(p)v(f)$.
\end{defn}

\begin{ex}
Let $M:=\R^n$ and $p\in\R^n$. Then each $v\in\R^n$ can be regarded as a tangent vector at $p$ as follows: 
$$v(f)=\frac{\dd}{\dd t}f(p+tv)\bigg|_{t=0}.$$
\end{ex}
Let $T_pM$ denote the space of tangent vectors at $p$. Then one can show that $T_pM$ is a vector space.

\begin{rem}
$\dim T_pM$ is given by the dimension of the manifold $M$.
\end{rem}

\begin{rem}
For $\varepsilon>0$, let $\gamma\colon(-\varepsilon,\varepsilon)\to M$ be a smooth curve such that $\gamma(0)=p$. Define $v_\gamma\colon C^\infty(M)\to \R$ by $v_\gamma(f)=\frac{\dd}{\dd t}(f\circ \gamma)(t)\big|_{t=0}$. Then $v_\gamma\in T_pM$. In fact, it can be shown that each $v\in T_pM$ appears in this way. This is the geometric way of thinking about a tangent vector at $p\in M$.
\end{rem}

\begin{defn}[Vector bundle]
Let $M$ be a smooth manifold. A \emph{real (complex) vector bundle} of rank $k$ is a pair $(E,\pi)$, where $E$ is a smooth manifold, $\pi\colon E\to M$ a smooth map, which is surjective such that
\begin{enumerate}
\item{For each $p\in M$, $E_p:=\pi^{-1}(p)$ is a real (complex) vector space,}
\item{For each $p\in M$, there is a neighborhood $U$ of $p$ and a diffeomorphism
$$\phi\colon \pi^{-1}(U)\to U\times \R^k (U\times \mathbb{C}^k),$$
such that the diagram
\[
\begin{tikzcd}
\pi^{-1}(U)\arrow[d,"\pi"]\arrow[r,"\phi"]&U\times \R^k (U\times \mathbb{C}^k)\arrow[dl,"\mathrm{pr}_1"]\\
U&
\end{tikzcd}
\]
is commutative and $\phi\big|_{E_p}\colon E_p\to \{p\}\times\R^k (\{p\}\times\mathbb{C}^k)$ is a linear isomorphism. We call $E_p$ the \emph{fiber} over $p$. Here we have denoted by $\mathrm{pr}_1$ the projection onto the first factor.
}
\end{enumerate}
\end{defn}

\begin{ex}[Trivial bundle]
Let $E:=M\times \R^k\xrightarrow{\pi}M$ be a vector bundle of rank $k$. Such a vector bundle is called a \emph{trivial vector bundle}.
\end{ex}
\begin{ex}[Tangent bundle]
Let $M$ be a smooth manifold of dimension $n$. Define $TM:=\bigsqcup_{p\in M}T_pM$. We call $TM$ the \emph{tangent bundle} of $M$.
\end{ex}

\begin{exe}
Show that the vector bundle $TM$ can be given a smooth structure so that it is a vector bundle over $M$ of rank $n$, where $n=\dim M$. 
\end{exe}

\begin{exe}[Cotangent bundle]
We can define the \emph{dual bundle} $T^*M:=\bigsqcup_{p\in M}(T_pM)^*$, where $(T_pM)^*$ is the dual of $T_pM$. Show that $T^*M$ is a vector bundle over $M$ of rank $n$ where $n=\dim M$. We call $T^*M$ the \emph{cotangent bundle} of $M$.
\end{exe}

\begin{rem}
Let $E$ and $F$ be two vector bundles over $M$. Then we can construct new vector bundles over $M$ as follows:
\begin{enumerate}
\item{(Tensor product bundle)
$E\otimes F\xrightarrow{\pi}$, where $(E\otimes F)_p:=E_p\otimes F_p$ for all $p\in M$,
}
\item{(Dual bundle)
$E^*\xrightarrow{\pi}M$, where $E_p^*:=(E_p)^*$ for all $p\in M$,
}
\item{(Direct sum of bundles)
$E\oplus F\xrightarrow{\pi}M$, where $(E\oplus F)_p:=E_p\oplus F_p$ for all $p\in M$,
}
\item{(Exterior bundle)
$\bigwedge^kE^*\xrightarrow{\pi} M$, where $(\bigwedge^kE)^*_p:=\bigwedge^k E_p^*$ for all $p\in M$.
}
\end{enumerate}
\end{rem}

\begin{defn}[Section]
Let $(E,\pi)$ be a vector bundle over $M$. A \emph{smooth section} $s$ of $E$ is a smooth map $s\colon M\to E$ such that $\pi\circ s=\id_M$. This means $s(p)\in E_p$ for all $p\in M$.
\end{defn}
\begin{rem}
We will use $\Gamma(M,E)$ to denote the space of all smooth sections on $M$. Note that $\Gamma(M,E)$ is a(n) (infinite-dimensional) vector space. 
\end{rem}

\begin{ex}
Let $E:=M\times \R^k\xrightarrow{\pi}M$. Then $\Gamma(M,E)$ can be identified with $C^\infty(M,\R^k)$, which are $\R^k$-valued smooth functions on $M$.
\end{ex}

\begin{defn}[Line bundle]
When $E$ is a vector bundle over $M$ of rank $1$, we call it a \emph{line bundle} over $M$.
\end{defn}

\begin{rem}
If $E:=M\times\R\xrightarrow{\pi}M$, then $\Gamma(M,E)=C^\infty(M)$.
\end{rem}

\begin{ex}
Let $TM$ be the tangent bundle over $M$. Then $\Gamma(M,TM)=\mathrm{Vect}(M)$, which is the space of vector fields on $M$.
\end{ex}

\begin{ex}[$1$-form]
Let $T^*M$ be the cotangent bundle over $M$. Then a section $\alpha\in\Gamma(M,T^*M)$ is called a \emph{differential $1$-form}.
\end{ex}

Let $X$ be a vector field on $M$ and $\alpha$ be a $1$-form on $M$. Then $X_p\in T_pM$ and $\alpha_p\in T_p^*M$. Hence, we can define a smooth function 
$$(\alpha(X))(p):=\alpha_p(X_p).$$
Moreover, if $f\in C^\infty(M)$, then $(f\alpha)(X)=f(\alpha(X))$ and $\alpha(fX)=f\alpha(X)$. In fact $\alpha\colon \Gamma(M,TM)\to C^\infty(M)$ is $C^\infty(M)$-linear, i.e. $\alpha(fX+gY)=f\alpha(X)+g\alpha(Y)$ for all $f,g\in C^\infty(M)$ and $X,Y\in\Gamma(M,TM)$. In fact, it can be shown that $$\Gamma(M,T^*M)=\{\alpha\colon \Gamma(M,TM)\to C^\infty(M)\mid \text{$\alpha$ is $C^\infty(M)$-linear}\}.$$ 

\begin{ex}
Let $f\in C^\infty(M)$ and define $\dd f$ by $\dd f(X):=X(f)$, where $X$ is a vector field on $M$. Then we can check that $\dd f\colon \Gamma(M,TM)\to C^\infty(M)$ and that it is $C^\infty(M)$-linear. Hence $\dd f$ defines a $1$-form on $M$.
\end{ex}

\begin{ex}[$k$-form]
Let $T^*M$ be the cotangent bundle over $M$. Consider the vector bundle $\bigwedge^k T^*M$. A $k$-\emph{form} on $M$ is a section of $\bigwedge^k T^*M$. Let $\alpha$ be a $k$-form on $M$. Then for each $p\in M$ we have $\alpha_p\in\bigwedge^k(T_pM)^*$. Let $X_1,...,X_k$ be vector fields on $M$. Then, we can define a smooth function $(\alpha(X_1,...,X_k))(p):=\alpha_p(X_1(p),...,X_k(p))$. One can check that $\alpha$ is $C^\infty(M)$ linea in each argument and alternating. Hence a $k$-form on $M$ can be regarded as a map 
$$\alpha\colon \underbrace{\Gamma(M,TM)\times\dotsm \times \Gamma(M,TM)}_{k}\to C^\infty(M),$$
which is $C^\infty(M)$-multilinear and alternating.
\end{ex}

\begin{ex}
Let $M:=\R^n$. Then we have $TM=\R^n\times\R^n\xrightarrow{\pi}\R^n$. Thus, $T^*M=\R^n\times(\R^n)^*\xrightarrow{\pi}\R^n$. A vector field $X$ on $\R^n$ can be written as 
$$X=\sum_{i=1}^nf_i\frac{\partial}{\partial x^{i}},$$
where $\frac{\partial}{\partial x^1},...,\frac{\partial}{\partial x^n}$ are the coordinate vector fields. Let $\dd x^{i}$ be the $1$-form dual to $\frac{\partial}{\partial x^{i}}$. Then, any $1$-form on $\R^n$ can be written as $\sum_{j=1}^n g_j\dd x^{j}$. A $k$-form on $\R^n$ can be represented as 
$$\sum_{1\leq i_1<\dotsm <i_k\leq n}g_{i_1,...,i_k}\dd x^{i_1}\land\dotsm \land \dd x^{i_k},$$
where $\land$ is the wedge product (alternating tensor product).
\end{ex}

\subsection{Exterior derivative}
Let $\Omega$ be a $k$-form on $M$. We will think of $\Omega$ as a map $$\underbrace{\Gamma(M,TM)\times\dotsm \times \Gamma(M,TM)}_{k}\to C^\infty(M),$$ which is $C^\infty(M)$-multilinear and alternating. Then $\dd \Omega$ is a $(k+1)$-form on $M$ defined by 
\begin{multline*}
\dd\Omega(X_1,...,X_{k+1})=\sum_{j=1}^{k+1}(-1)^{j+1}X_j(\Omega(X_1,...,\widehat{X}_j,...,X_{k+1}))+\\\sum_{1\leq i<j\leq k+1}(-1)^{i+j}\Omega([X_i,X_j],X_1,...,\widehat{X}_i,...,\widehat{X}_j,...,X_{k+1}).
\end{multline*}

\begin{ex}
Let $f\in C^\infty(M)$. Then $\dd f(X)=X(f)$.
\end{ex}

\begin{ex}
Let $M:=\R^2$ and $\alpha:=p\dd x$. Denote by $\partial_x$ the tangent vector $\frac{\partial}{\partial x}$.
Then $\dd\alpha(\partial_{x},\partial_x)=\dd\alpha(\partial_p,\partial_p)=0$ and $\dd\alpha(\partial_x,\partial_p)=1$. Moreover, $\dd p\land\dd x(\partial_x,\partial_x)=\dd p\land\dd x(\partial_p,\partial_p)=0$ and $\dd p\land\dd x(\partial_x,\partial_p)=1$. Thus we get $\dd\alpha=\dd p\land\dd x$.
\end{ex}

\begin{ex}
Let $M:=\R^n\ni(x^1,...,x^n)$. If $\alpha:=f\dd x^1\land\dotsm \land \dd x^k$, then $$\dd\alpha=\sum_{j=k+1}^n\frac{\partial f}{\partial x^j}\dd x^j\land \dd x^1\land\dotsm\land \dd x^k.$$ More generally, if $\alpha:=\sum_{i_1<\dotsm <i_k}f\dd x^{i_1}\land\dotsm \land \dd x^{i_k}$, then $$\dd\alpha=\sum_{j\not\in\{i_1,...,i_k\}}\frac{\partial f}{\partial x^j}\dd x^j\land \dd x^{i_1}\land\dotsm\land \dd x^{i_k}.$$
\end{ex}

\begin{rem}
The operator $\dd$ has the following properties:
\begin{enumerate}
\item{$\dd$ is $\R$-linear,}
\item{If $\omega$ is a $k$-form and $\eta$ is an $\ell$-form on $M$, then $$\dd(\omega\land\eta)=\dd\omega\land\eta+(-1)^k\omega\land \dd\eta,$$
}
\item{$\dd\circ \dd=0$,}
\item{$\dd f(X)=X(f)$ for all $f\in C^\infty(M)$.}
\end{enumerate}
Moreover, these properties uniquely determine $\dd$ on $k$-forms for $0\leq k\leq \dim M$.
\end{rem}

\subsection{Exercises}

\begin{exe}
Let $M:=\R$ and $\calA:=\{\R,\id_\R\}$ and $\calA':=\{(\R,\phi\colon\R\to\R,\phi(x)=x^3)\}$ 
\begin{enumerate}
\item{Show that $\calA'$ is a smooth atlas on $\R$.
}
\item{Show that $\calA$ and $\calA'$ induce different smooth structures on $\R$.
}
\begin{defn}[Standard structure]
The smooth structure on $\R$ induced by $\calA$ (i.e. the smooth maximal atlas containing $\calA$) is called \emph{standard smooth structure} on $\R$.
\end{defn}
\item{
Define $f\colon (\R,\calA)\to(\R,\calA')$ by $f(x)=x^3$. Show that $f$ is a diffeomorphism from $\R$ with the standard smooth structure to $\R$ with the smooth structure induced by $\calA'$.

}
\end{enumerate}
\end{exe}

\begin{exe}
Let $M$ be a smooth manifold. 
\begin{enumerate}
\item{Let $X$ and $Y$ be vector fields on $M$ and $f,g\in C^\infty(M,\R)$. Show that $[fX,gY]=fX(g)Y-gY(f)X+fg[X,Y]$.
}
\item{Show that $[\enspace,\enspace]$ satisfies the Jacobi identity.}
\end{enumerate}
\end{exe}

\begin{exe}
Let $M$ be a smooth manifold with $\dim M=n$.
\begin{enumerate}
\item{Let $\gamma\colon (-1,1)\to M$ be a smooth map (here $(-1,1)$ is given the standard smooth structure) such that $\gamma(0)=p$. Let $f\in C^\infty(M,\R)$. Define $v_\gamma f=\frac{\dd}{\dd t}(f(\gamma(t)))\big|_{t=0}$. Show that $v_\gamma\in T_pM$.
}
\item{Let $p\in M$. Let $(U,\phi)$ be a coordinate chart at $p$ such that $\phi(p)=0$. Let $w\in\R^n$ and $\varepsilon>0$ small such that $tw\in\phi(U)$ for all $t\in (-\varepsilon,\varepsilon)$. Define $\alpha(t)=tw$, $\alpha\colon (-\varepsilon,\varepsilon)\to\R^n$ and $\gamma(t)=\phi^{-1}(\alpha(t))$. Then $\gamma\colon(-\varepsilon,\varepsilon)\to M$ is smooth. Show that $\dd\phi_p(v_\gamma)=w$. \emph{Hint: Use that
$$w(f):=\frac{\dd}{\dd t}(f(tw))\big|_{t=0},$$
for $f\in C^\infty(\R^n)$ and $w\in\R^n$.
} 
}
\end{enumerate}
\end{exe}

\begin{exe}
Let $M$ and $N$ be smooth manifolds and $F\colon M\to N$ be a diffeomorphism. Let $p\in M$. Show that $\dd F_p\colon T_pM\to T_{F(p)}N$ is a vector space isomorphism. (This exercise implies that $M$ and $N$ are diffeomorphic and thus $\dim M=\dim N$.)
\end{exe}

\begin{exe}
Let $M$ and $N$ be smooth manifolds and $F\colon M\to N$ be a smooth map. 
\begin{enumerate}
\item{Let $\omega$ be a $k$-form on $N$. Given vector fields $X_1,...,X_k$ in $M$, define 
$$((F^*\omega)(X_1,...,X_k))(p):=\omega(\dd F_p(X_1),...,\dd F_p(X_k)).$$
Show that $F^*\omega$ is a $k$-form on $M$. (This exercise shows that we can pull back differential forms.)
}
\item{Show that $F^*(\dd \omega)=\dd(F^*\omega)$ for any $k$-form on $N$ with $k\in\N$. (This exercise shows that $\dd$ commutes with the pullback.)
}
\end{enumerate}
\end{exe}

\begin{exe}
Let $V$ be a finite dimensional vector space and $v\in V$. For $\alpha\in\bigwedge^kV^*$, define 
$$\iota_v\alpha(v_2,...,v_k):=\alpha(v,v_2,...,v_k).$$
\begin{enumerate}
\item{Show that $\alpha\in\bigwedge^k V^*$ implies that $\iota_v\alpha\in \bigwedge^{k-1}V^*$. Hence, conclude that $\iota_v$ defines a linear map $\iota_v\colon \bigwedge^kV^*\to \bigwedge^{k-1}V^*$.
}
\item{Show that $\iota_v\circ \iota_v=0$. \emph{Hint: think of a very simple fact about $\alpha$.}
}
\end{enumerate}
\end{exe}

\section{Symplectic manifolds and Hamiltonian systems}
See also \cite{daSilva01} for more on symplectic geometry and its relation to Hamiltonian mechanics.

\subsection{Symplectic manifolds}

\begin{defn}[Closed/exact]
We call a $k$-form $\omega$ \emph{closed}, if $\dd\omega=0$. It is called \emph{exact} if there is a $(k-1)$-form $\alpha$ such that $\dd\alpha=\omega$.
\end{defn}

\begin{ex}
If $\omega$ is exact, then $\dd\omega=0$, i.e. exact forms are closed as well. Let $M=\R^n$, then $\omega$ is closed if and only if $\omega$ is exact (this is given by the \emph{Poincar\'e lemma}).
\end{ex}

\begin{defn}[Symplectic manifold]
A \emph{symplectic manifold} is a pair $(M,\Omega)$, where $M$ is a smooth manifold and $\Omega$ is a $2$-form on $M$ such that
\begin{enumerate}
\item{$\Omega$ is closed, i.e. $\dd\Omega=0$,}
\item{$\Omega$ is nondegenerate, i.e. for all $q\in M$, $\Omega^\flat\colon T_qM\to T_q^*M$ is injective.}
\end{enumerate}
\end{defn}

\begin{defn}[Tautological $1$-form]
Let $M:=T^*N$ for some manifold $N$. Define a $1$-form $\alpha$ on $M$ as 
$$\alpha_{x,p}(X_{x,p}):=\pi^N(\dd \pi^M_{x,p}X_{x,p}),$$
where $\pi^N\colon T^*N\to N$, $\pi^M\colon TM\to M$ and $X_{x,p}\in T_{x,p}M$. The form $\alpha$ is called the \emph{tautological $1$-form} on $T^*N$.
\end{defn}

\begin{ex}
Let $M:=T^*\R\cong \R\times \R\ni(x,p)$. Let $\alpha:=f\dd x+g\dd p$. Then $\alpha(\partial_x)=f$ and $\alpha(\partial_p)=0$, thus $\alpha=f\dd x$. on the other hand $\alpha_{x,p}(\partial_x)=p$ and hence $\alpha=p\dd x$. More generally, if $M:=T^*\R^n\ni(x^{1},...,x^n,p_1,...,p_n)$, then $\alpha=\sum_{1\leq j\leq n}p_j\dd x^j$
\end{ex}

\begin{exe}
Let $(U,\phi)$ be a local coordinate system on $M=T^*N$ given by $$\phi(q)=(x^1(q),...,x^n(q),p_1(q),...,p_n(q)),$$ Show that $\alpha=\sum_{1\leq j\leq n}p_j\dd x^j$. Moreover, show that $(T^*N,\Omega=\dd\alpha)$ is a symplectic manifold.
\end{exe}

\subsection{The Lie derivative}

\begin{defn}[Lie derivative]
Let $f\in C^\infty(M)$ and $X$ be a vector field. The Lie derivative of $f$ along $X$ is defined as $\emph{L}_Xf=X(f)$. Let $X$ and $Y$ be two vector fields. Then we define $\emph{L}_XY=[X,Y]$. Moreover, let $X$ be a vector filed and $\alpha$ a $1$-form. Then $\emph{L}_X\alpha$ is a $1$-form defined by the equation $$\emph{L}_X(\alpha(Y))=(\emph{L}_X\alpha)(Y)+\alpha(\emph{L}_XY).$$ More generally, if $\alpha$ is a $k$-form  then $\emph{L}_X\alpha$ is again a $k$-form defined by $$(\emph{L}_X\alpha)(Y_1,...,Y_k):=\emph{L}_X(\alpha(Y_1,...,Y_k))-\sum_{j=1}^k\alpha(Y_1,...,Y_{j-1},[X,Y_j],Y_{j+1},...,Y_k).$$
\end{defn}

\begin{rem}
Given a $k$-form $\alpha$, and a vector field $X$, $(\emph{L}_X\alpha)(p)$ is the rate of change of $\alpha$ in the direction of the so-called \emph{flow} of $x$ at $p$.
\end{rem}

\begin{exe}
Given a vector field $X$ and a $k$-form $\alpha$, $\iota_X\alpha$ is a $(k-1)$-form defined by $$(\iota_X\alpha)(Y_1,...,Y_{k-1})=\alpha(X,Y_1,...,Y_{k-1}).$$
Show that
\begin{enumerate}
\item{$\iota_X\circ\iota_X=0$,}
\item{$\emph{L}_X=\dd\circ \iota_X+\iota_X\circ\dd$ (Cartan's magic formula).}
\end{enumerate}
\end{exe}

\begin{rem}
We denote by $\Omega^k(M)$ the space of global $k$-forms on $M$.
\end{rem}

\begin{defn}[Poisson bracket II]
\label{PBII}
Let $(M,\Omega)$ be a symplectic manifold and let $f\in C^\infty(M)$. Then $\dd f\in\Omega^1(M)$. Note that $\dd f$ defines a vector field $X_f$ on $M$ by $(\Omega^\flat)^{-1}(X_f)=\dd f$, i.e. $\Omega(X_f,\enspace)=-\dd f$ or equivalently $\iota_{X_f}\Omega=-\dd f$. Moreover, note that $\dd(\iota_{X_f}\Omega)=-\dd(\dd f)=0$. We can define a Poisson bracket for $f,g\in C^\infty(M)$ by 
$$\{ f,g\}:=\Omega(X_f,X_g)=-\dd f(X_g)=-X_g(f)=X_f(g).$$
\end{defn}

\begin{exe}
Let $M:=T^*\R^n$ together with $\Omega=\dd\alpha$, where $\alpha$ is the canonical $1$-form on $T^*\R^n$. Then 
\begin{equation}
\label{eq:Poisson_bracket}
\{f,g\}=\sum_{j=1}^n\left(\frac{\partial f}{\partial x^j}\frac{\partial g}{\partial p_j}-\frac{\partial g}{\partial x^j}\frac{\partial f}{\partial p_j}\right).
\end{equation}
\end{exe}

\begin{prop}[Properties of the Poisson bracket]
We have that $\{\enspace,\enspace\}$, as defined in \eqref{eq:Poisson_bracket}, is $\R$-bilinear, antisymmetric and satisfies the Jacobi identity.
\end{prop}

\begin{proof}
It is straightforward to show $\R$-bilinearity and antisymmetry. We want to show first that 
\begin{equation}
\label{VF1}
X_{\{f,g\}}=[X_f,X_g].
\end{equation}
To show \eqref{VF1}, we will show that 
\begin{equation}
\label{SF1}
\Omega(Y,X_{\{f,g\}})=\Omega(Y,[X_f,X_g]),
\end{equation}
for all $Y\in \mathrm{Vect}(M)$. As $\Omega$ is nondegenerate, \eqref{VF1} implies \eqref{SF1}. Note that $\emph{L}_{X_f}\Omega=\dd(\iota_{X_f}\Omega)+\iota_{X_f}(\dd\Omega)=0$. Hence, for all $Y\in \mathrm{Vect}(M)$
\[
0=(\emph{L}_{X_f}\Omega)(Y,X_g)=\emph{L}_{X_f}(\Omega(Y,X_g))-\Omega([X_f,Y],X_g)=\Omega(Y,[X_f,X_g]).
\]
Hence we get 
\begin{align*}
\Omega(Y,[X_f,X_g])&=X_f(\Omega(Y,X_g))-\Omega([X_f,Y],X_g)\\
&=X_f(Y(g))-[X_f,Y](g)=X_f(Y(g))-X_f(Y(g))+Y(X_f(g))=Y(X_f(g))\\
&=-Y(X_g(f))=Y(\{f,g\})=\dd(\{f,g\})(Y)=-\dd(\{f,g\})(-Y)\\
&=\Omega(Y,X_{\{f,g\}}).
\end{align*}
Now let $h\in C^\infty(M,\R)$. Using \eqref{VF1}, we get 
\[
X_{\{f,g\}}(h)=[X_f,X_g](h)
\]
and thus $\{\{f,g\},h\}=\{f,\{g,h\}\}-\{g,\{f,h\}\}$, and hence $\{f,\{g,h\}\}=\{\{f,g\},h\}+\{g,\{f,h\}\}$.
\end{proof}

\subsection{Hamiltonian systems}
\begin{defn}[Hamiltonian system]
A \emph{Hamiltonian system} is a triple $(M,\Omega, H)$, where $(M,\Omega)$ is a symplectic manifold and $H\colon M\to\R$ is a smooth function.
\end{defn}

Let $X_H$ be the Hamiltonian vector field associated to $M$. The integral curves of $X_H$ are trajectories of motions. In a local coordinate system, computation of integral curves of $X_H$ boils down to Hamilton's equations. Let $\gamma(t)$ be an integral curve of $X_H$. Then for any $f\in C^\infty(M,\R)$ we have $X_H(f)(\gamma(t))=\frac{\dd}{\dd t}f(\gamma(t))$. This implies 
\[
\frac{\dd}{\dd t}f(\gamma(t))=X_H(f)(\gamma(t))=\{H,f\}(\gamma(t)).
\]
($f$ is conserved along $\gamma$ if and only if $\{H,f\}=0$ along $\gamma$).

\subsection{Short summary}
We want to give a short summary of this section:
\begin{itemize}
\item{The phase space (or state space) of Classical Mechanics leads to the notion of a symplectc manifold.
}
\item{A classical observable is a function on the phase space. A particular choice of an observable corresponds to a \emph{physical system}.
}
\item{Conservation can be expressed using the Poisson bracket.
}
\item{Let $(M,\Omega)$ be a symplectic manifold. Then $(C^\infty(M,\R),\{\enspace,\enspace\})$ is a Lie algebra.
}
\end{itemize}

\subsection{Exercises}

\begin{exe}
Let $M:=\R^n$ and $(x^1,...,x^n)$ be the global coordinates on $\R^n$. Let $\alpha$ be a $k$-form and $\beta$ be an $\ell$-form on $M$. 
Show that $\dd(\alpha\land\beta)=\dd\alpha\land \beta+(-1)^k\alpha\land \beta$.
\end{exe}

\begin{exe}
\label{ex_mnf}
Let $X$ be a vector field on $M$, $\alpha$ a $k$-form and $\beta$ an $\ell$-form on $M$. Show that $\emph{L}_X(\alpha\land\beta)=(\emph{L}_X\alpha)\land \beta+\alpha\land\emph{L}_X\beta$.
\end{exe}

\begin{exe}[Liouville's theorem]
Let $(M,\Omega)$ be a symplectic manifold with $\dim M=2n$. Define $\lambda=\frac{1}{n!}(\underbrace{\Omega\land\dotsm\land\Omega}_{n})$ Let $f\in C^\infty(M)$ and $X_f$ be the Hamiltonian vector field. Use Exercise \ref{ex_mnf} to show that $\emph{L}_{X_f}\lambda=0$. (This statement is called \emph{Liouville's theorem})
\end{exe}

\begin{exe}
Let $(M,\Omega)$ be a symplectic manifold. Show that $\{f,gh\}=\{f,g\}h+g\{h,f\}$ for all $f,g,h\in C^\infty(M)$, where $\{\enspace,\enspace\}$ is the Poisson bracket.
\end{exe}

\begin{exe}
Let $M=\R^2\ni (x,p)$. Let $\alpha=p\dd x$. Compute $\iota_X\alpha$, where $X=f\partial_x+g\partial_p$ and $\iota_{\partial_x}\omega$, $\iota_{\partial_p}\omega$, where $\omega=\dd\alpha$.
\end{exe}

\section{Introduction to Quantum Mechanics}

\subsection{Failure of Classical Mechanics}
We want to look at the lifespan of a \emph{Hydrogen atom}. Consider a positive charge $e^+$ (proton) sitting in the center of a circle with radius $r$ and a negative charge $e^-$ (electron) moving along the circle trajectory with velocity $\vec{v}$ (pointing to the direction tangential to the circle). Thus, we have an acceleration $\vec{a}$ on $e^-$ pointing to the center (perpendicular to $\vec{v}$), which comes from the centripetal force, given by $\vert\vec{a}\vert=\frac{v^2}{r}$. Here $v^2=\vec{v}\cdot\vec{v}$. Moreover we have a potential $V(r)=-\frac{e^2}{r}$ (here $e$ is the absolute value of the charge, i.e. $e=\vert e^+\vert=\vert e^-\vert$). This potential is called \emph{Coulomb's law}. Moreover, let $E$ denote the total energy, i.e. $E=\frac{1}{2}mv^2+V(r)$. Now since the electron is coupled to the electromagnetic field, it produces electromagnetic waves which carries energy away. Hence we get
$\frac{\dd E}{\dd t}=-\frac{e^2\vert \vec{a}\vert^2}{C}$, where $C>0$ is some constant. Now suppose that Newton's second law holds. Then $\vert \vec{F}\vert=m\vert \vec{a}\vert$, hence $\vert \vec{a}\vert=\frac{1}{m}\vert\vec{F}\vert=\frac{e^2}{mr^2}$, and thus $mv^2=\frac{e^2}{r}$. This gives us $E=-\frac{1}{2}\frac{e^2}{r}$ and hence $\frac{\dd E}{\dd t}=\frac{e^2}{2r^2}\frac{\dd r}{\dd t}$. Using $\frac{\dd r}{\dd t}=-\frac{e^4}{C}\frac{1}{r^2}$, we get that $r(t)$ is rapidly decreasing. In fact, it can be shown that $r\to 0$ in a very short time. This shows also that the Hydrogen atom collapses in a short time, which in fact does not coincide with the experiments. 

\vspace{0.3cm}
\emph{{\bf Upshot:} Classical Mechanics does not fully explain the behaviour of atomic particles.}
\vspace{0.3cm}

\subsection{Axioms of Quantum Mechanics}
The axioms of Quantum Mechanics are motivated by the following experimental facts:
\begin{itemize}
\item{Objects are observed to have wave-like and particle-like behaviour (\emph{wave-particle duality}).
}
\item{We can only predict the \emph{probabilities} of an outcome.}
\end{itemize}

\begin{rem}
We have the notion of a \emph{wave function}: A wave function $\psi$ is a \emph{function} of $x\in\R^n$ , which we interpret as describing the possible values of the position of a particle and it \emph{evolves} in time obeying a \emph{wave-like} equation.
\end{rem}

\subsubsection{Digression: complex Hilbert space, self-adjoint operators} We want to give some mathematical tools for the understanding of the quantum theory. 

\begin{defn}[Complex inner product space]
A \emph{complex inner product space} is a pair $(\calH,\langle\enspace,\enspace\rangle)$, where $\calH$ is a complex vector space and the map $\langle\enspace,\enspace\rangle\colon \calH\times\calH\to \mathbb{C}$ is such that for all $\phi,\psi,\phi_1,\phi_2,\psi_1,\psi_2\in\calH$ and $c\in\mathbb{C}$
\begin{enumerate}
\item{$\langle c\phi,\psi\rangle=\bar c\langle \phi,\psi\rangle$, and $\langle\phi,c\psi\rangle=c\langle \phi,\psi\rangle$,
}
\item{$\langle \phi,\psi_1+\psi_2\rangle=\langle\phi,\psi_1\rangle+\langle\phi,\psi_2\rangle$, and $\langle \phi_1+\phi_2,\psi\rangle=\langle\phi_1,\psi\rangle+\langle\phi_1,\psi\rangle$,
}
\item{$\langle\phi,\psi\rangle=\overline{\langle\psi,\phi\rangle}$ (Hermitian)
}
\item{$\langle\phi,\phi\rangle\geq0$ and $\langle\phi,\phi\rangle=0$ if and only $\phi=0$.
}
\end{enumerate}
We call $\langle\enspace,\enspace\rangle$ a \emph{complex inner product}.
\end{defn}
Define $\|\phi\|:=\langle\phi,\phi\rangle^{1/2}$. Let $\{\phi_n\}_n$ be a sequence in $\calH$, we say $\{\phi_n\}_n$ is \emph{Cauchy} if $\|\phi_n-\phi_m\|\to 0$ as $n,m\to \infty$. Moreover, we say $(\calH,\langle\enspace,\enspace\rangle)$ is \emph{complete} if every Cauchy sequence converges in $\calH$, i.e. $\{\phi_n\}_n$ is Cauchy implies there is some $\phi\in \calH$ such that $\|\phi_n-\phi\|\to 0$ as $n\to\infty$.

\begin{defn}[Complex Hilbert space]
A \emph{complex Hilbert space} is a complete complex inner product space.
\end{defn}

\begin{ex}
Take $\calH:=\mathbb{C}^n$ with inner product $\langle z,w\rangle:=\sum_{j=1}^n\bar z_jw_j$.
\end{ex}

\begin{ex}
Take $\calH:=L^2(\R^n,\dd x)$ with inner product 
\[
\langle f,g\rangle:=\int_{\R^n}\overline{f(x)}g(x)\dd x.
\]
\end{ex}

\begin{ex}
Let $(X,\calB,\mu)$ be a measure space, i.e. $X$ is a set, $\calB$ a $\sigma$-algebra of subsets of $X$, and $\mu$ a measure. Then $L^2(X,\mu)$ is a Hilbert space, where 
\[
\langle f,g\rangle:=\int_X\overline{f(x)}g(x)\dd\mu(x).
\]
\end{ex}

\begin{defn}[Operator]
Let $\calH$ be a Hilbert space. An \emph{operator} on $\calH$ is a pair $(A,\mathrm{Dom}(A))$, where $\mathrm{Dom}(A)$ is a dense subspace of $\calH$, called the \emph{domain} of $A$, and $A\colon \mathrm{Dom}(A)\to \calH$ is linear. $A$ is \emph{bounded} if there is some $c>0$ such that for all $\phi\in \mathrm{Dom}(A)$, $\|A\phi\|\leq c\|\phi\|$.
\end{defn}

\begin{rem}
If $A$ is bounded, then the denseness of $\mathrm{Dom}(A)$ implies that it can be extended to a linear map $A\colon \calH\to \calH$. Moreover $\| A\phi\|\leq c\|\phi\|$ for all $\phi\in\calH$.
\end{rem}

\begin{rem}
Given an operator $(A,\mathrm{Dom}(A))$, there is an operator 
\[
(A^*,\mathrm{Dom}(A^*))
\]
such that $\langle  A\phi,\psi\rangle=\langle\phi,A^*\psi\rangle$ for all $\phi\in \mathrm{Dom}(A)$ and for all $\psi\in \mathrm{Dom}(A^*)$.
\end{rem}

\begin{defn}[Adjoint]
The operator $A^*$ is called the \emph{adjoint} of $A$
\end{defn}

\begin{defn}[Symmetric]
An operator $A$ is called \emph{symmetric} if for all $\phi,\psi\in \mathrm{Dom}(A)$
$$\langle A\phi,\psi\rangle=\langle \phi,A\psi\rangle.$$
\end{defn}

\begin{defn}[Self-adjoint]
An operator $A$ is called \emph{self-adjoint} if $\mathrm{Dom}(A)=\mathrm{Dom}(A^*)$ and $A^*\phi=A\phi$ for all $\phi\in \mathrm{Dom}(A)$.
\end{defn}

\begin{defn}[Resolvent]
Let $A$ be an operator on $\calH$ and let $\lambda\in\mathbb{C}$. We say that $\lambda$ is in the \emph{resolvent set} $\rho(A)$ if 
$$(A-\lambda I)\colon \mathrm{Dom}(A)\to \calH$$
is a bijection and $(A-\lambda I)^{-1}$ is bounded. Here $I$ is the \emph{identity operator} on $\calH$.
\end{defn}

\begin{defn}[Specturm]
The \emph{spectrum} $\sigma(A)$ of an operator $A$ is defined by 
\[
\sigma(A):=\mathbb{C}\setminus\rho(A).
\]
\end{defn}

\begin{ex}[Eigenvalue]
Let $A$ be an operator on $\calH$ and let $\lambda\in\mathbb{C}$. Assume that there is some $\psi\not=0$ in $\calH$ such that $A\psi=\lambda\psi$. Then $\lambda\in\sigma(A)$, since $(A-\lambda I)^{-1}$ does not exist. Such a $\lambda$ is called an \emph{eigenvalue} of $A$.
\end{ex}

\subsubsection{Axioms} We can now formulate the \emph{axioms} of Quantum Mechanics. 
\begin{enumerate}[(QM1)]
\item{To every quantum system, there is an associated infinite-dimensional separable complex Hilbert space $\calH$, called the \emph{space of states}. The \emph{pure} state of a system is represented by a unit vector in $\calH$. Let $\phi_1$ and $\phi_2$ be two unit vectors in $\calH$ such that $\phi_1=c\phi_2$ for some $c\in\mathbb{C}$. Then $\phi_1$ and $\phi_2$ represent the same physical state. Consider the set $\calS:=\{\psi\in\calH\mid\|\psi\|=1\}$. Given $\phi,\psi\in\calS$, we have $$\vert\langle\phi,\psi\rangle\vert^2\leq \|\phi\|^2\|\psi\|^2=1.$$
Here $\vert\langle\phi,\psi\rangle\vert^2$ can be interpreted as the probability of a physical system at $\phi$ given the physical system at $\psi$.
}
\item{An observable of a quantum system with the space of states given by $\calH$ is a self-adjoint operator on $\calH$. We define 
\[
\calA:=\{\text{self-adjoint operators on $\calH$}\}.
\]
}
\item{The process of measurement corresponds to the map 
\begin{align*}
\calA\times\calS&\to P(\R):=\{\text{probability measures on $\R$}\}\\
(A,\psi)&\mapsto \mu_A^\psi.
\end{align*}
Given $E\subseteq\R$ measurable (more precisely Borel measurable), $\mu_A^\psi(E)$ is interpreted as the probability of the \emph{measurement} of $A$ in the \emph{state} $\psi$ that is in $E$. Moreover, the \emph{expectation} of $A\in\calA$ in the state $\psi\in\calS$ is given by 
$$\langle A\rangle_\psi:=\int_\R\lambda \dd\mu_A^\psi(\lambda).$$
}
\item{The dynamics of a quantum system is governed by the \emph{Schr\"odinger equation}, i.e. there is a distinguished quantum observable $\widehat{H}$, such that the \emph{time evolution} $\psi(t)$ with $\psi(0)=\psi$ satisfies
$$\I\hbar\frac{\dd\psi(t)}{\dd t}=\widehat{H}\psi(t).$$
\begin{rem}
In the so-called \emph{Heisenberg picture} of Quantum Mechanics, the dynamics is governed by the equation $$\I\hbar \frac{\dd A(t)}{\dd t}=-[\widehat{H},A(t)],$$
where $A(0)=A$ and $[A,B]=AB-BA$ is the \emph{commutator} of operators.
\end{rem}
}
\end{enumerate}

\begin{ex}[Free particle in position space]
Consider a free particle moving in $\R^n$. Recall that the phase space is given by $M=T^*\R^n$ and the energy is $E(x,p)=\frac{1}{2m}p^2$. Then $\calH=L^2(\R^n,\dd x)$ (space of wave functions) and $\widehat{x}^j(f(x))=x^jf(x)$. Moreover, $\widehat{p}_j(f(x))=\I\hbar\frac{\partial f}{\partial x^j}$, and hence 
$$\widehat{H}=\sum_{j=1}^n\frac{1}{2m}\widehat{p}_j^2=-\frac{\hbar^2}{2m}\sum_{j=1}^n\frac{\partial^2}{\partial (x^j)^2}.$$
\end{ex}

\begin{ex}[Free particle in momentum space]
Consider a free particle moving in $\R^n$. Then $\calH=L^2(\R^n,\dd p)$ (space of wave functions) and $\widehat{x}^j(f(x))=-\I\hbar\frac{\partial f}{\partial p_j}$. Moreover, $\widehat{p}_j(f(p))=p_jf(p)$, and hence 
$$\widehat{H}=\sum_{j=1}^n\frac{1}{2m}\widehat{p}_j^2$$
\end{ex}

\begin{rem}
Starting from a classical mechanical system, we want to construct a quantum mechanical system. It turns out that one can construct many quantum mechanical systems from the same classical mechanical system as suggested by the examples above. We would like to understand how to \emph{compare} them.
\end{rem}

\section{Quantization}
We want to be able to pass from a classical to a corresponding quantum system. This is encoded in a \emph{Quantization map} $\mathscr{Q}$, i.e.
\[
\text{Classical Mechanics}\xrightarrow{\mathscr{Q}}\text{Quantum Mechanics}
\]
The classical state space is given by a symplectic manifold $(M,\Omega)$, whereas the quantum state space is given by a Hilbert space $\calH$, hence $\mathscr{Q}((M,\Omega))$ has to be a Hilbert space. The classical observables are given by smooth functions $f\in C^\infty(M,\R)$, whereas the quantum observables are given by self-adjoint operators, thus $\mathscr{Q}(f)$ will be a self-adjoint operator. Classical time evolution is given, for some Hamiltonian function $H\in C^\infty(M,\R)$, by the equation $\frac{\dd f}{\dd t}=\{H,f\}$ along the flow of $X_H$, whereas on the quantum time evolution is given, for a Hamiltonian self-adjoint operator $\widehat{H}$ on $\calH$, by the equation $\I\hbar\frac{\dd A(t)}{\dd t}=-[\widehat{H},A(t)]$ or equivalently $\frac{\dd A(t)}{\dd t}=\frac{\I}{\hbar}[\widehat{H},A(t)]$. This shows that the image of the Poisson bracket under $\mathscr{Q}$ will be given by the commutator $\frac{\I}{\hbar}[\enspace,\enspace]$. 

\begin{defn}[Quantization]
\emph{Quantization} of a classical mechanical system \emph{roughly} means the construction of a quantum mechanical system, starting from a classical mechanical system. Ideally, we want a procedure $\mathscr{Q}$ that assigns to a symplectic manifold $(M,\Omega)$ a separable Hilbert space, and to a smooth function $f\in C^\infty(M,\R)$ a self-adjoint operator $\mathscr{Q}(f)$ such that 
\begin{enumerate}[$(i)$]
\item{$\mathscr{Q}$ is linear in $f$,}
\item{$\mathscr{Q}(1)=\id_\calH$,}
\item{$\mathscr{Q}(\{f,g\})=\frac{\I}{\hbar}[\mathscr{Q}(f),\mathscr{Q}(g)]$.}
\end{enumerate}
Moreover, we want $\calH$ to be \emph{minimal}.
\end{defn}

\begin{rem}
An ideal quantization procedure does not exist (see \emph{Groenewold's thoerem} \cite{Groenewold1946}). In practice, we do not look for an ideal $\mathscr{Q}$.
\end{rem}

\subsection{Quantization of $T^*\R^n$ and ordering ambiguity}
Consider the quantum state space $\calH=L^2(\R^n,\dd x)$ (given in position space representation). Recall that we have position and momentum operators $\widehat{x}^j$ and $\widehat{p}_k$. Let $f\in C^\infty(T^*\R^n)$, such that $f(x,p)=x^jp_k$. Note that in Classical Mechanics, $x^jp_k=p_kx^j$. Define $\mathscr{Q}(x^j)=\widehat{x}^j$ and $\mathscr{Q}(p_k)=\widehat{p}_k$. Now, there are many choices to define $\mathscr{Q}(f)$. For example, we could take $\widehat{x}^j\widehat{p}_k$ or $\widehat{p}_k\widehat{x}^j$ or $\frac{\widehat{x}^j\widehat{p}_k+\widehat{p}_k\widehat{x}^j}{2}$. All these possibilities are different. More generally, if $f$ is a comlicated function, it is not clear how to define $\mathscr{Q}(f)$. This is called \emph{ordering ambiguity}. There are practical solutions to this problem such as \emph{Wick-ordered quantization} or \emph{Weyl quantization}, which depends on certain choices.

\subsection{Geometric Quantization}
Geometric quantization is roughly a quantization procedure that uses the data of symplectic geometry of a classical mechanical system and constructs a quantum mechanical system. There are two steps into the process: 
\begin{enumerate}[(Step1)]
\item{\emph{Prequantization}: construct a Hilbert space (called \emph{prequantum Hilbert space}) and a prequantized observable $\mathscr{Q}_{pre}(f)$, for $f\in C^\infty(M)$.
}
\item{\emph{Correction}: Get the quantum Hilbert space $\calH$ and the quantum observable $\mathscr{Q}(f)$ for $f\in C^\infty(M)$.
}
\end{enumerate}

\subsubsection{Prequantization of $T^*\R^n$}
We will construct a Hilbert space $\calH$ and an operator $\mathscr{Q}_{pre}(f)$ for $f\in C^\infty(T^*\R^n)$ such that $\mathscr{Q}_{pre}(1)=\id$ and $\mathscr{Q}_{pre}(\{f,g\})=\frac{\I}{\hbar}[\mathscr{Q}_{pre}(f),\mathscr{Q}_{pre}(g)]$. First, we can recall that $\{ x^k,p_j\}=\delta_{jk}\cdot 1$, and thus $$[\mathscr{Q}_{pre}(x^k),\mathscr{Q}_{pre}(p_j)]=-\I\hbar I\delta_{jk},$$ 
where $I$ denotes the identity operator on the prequantum Hilbert space. In particular, we have
\begin{equation}
\label{prequantum}
[\mathscr{Q}_{pre}(x^k),\mathscr{Q}_{pre}(p_k)]=-\I\hbar I.
\end{equation}
If $\calH$ is a Hilbert space such that $\mathscr{Q}_{pre}(x^k)$ and $\mathscr{Q}_{pre}(p_k)$ are two operators on $\calH$ such that \eqref{prequantum} holds, then $\calH$ must be infinite-dimensional. A natural choice for $\calH$ will be $L^2(\R^{2n})$ (prequantum Hilbert space). For the construction of the operators, we start with a first attempt by setting $\mathscr{Q}_{pre}(f):=-\I\hbar X_f$. Then 
\[
\frac{\I}{\hbar}[\mathscr{Q}_{pre}(f),\mathscr{Q}_{pre}(g)]=-\I\hbar[X_f,X_g]=-\I\hbar X_{\{f,g\}}=\mathscr{Q}_{pre}(\{f,g\}).
\]
The problem is that $\mathscr{Q}_{pre}(1)=0$, since $X_{f=1}=0$. The second attempt is to set $\mathscr{Q}_{pre}(f)=-\I\hbar X_f+f$. Then $\mathscr{Q}_{pre}(1)=1$, but $\frac{\I}{\hbar}[\mathscr{Q}_{pre}(f),\mathscr{Q}_{pre}(g)]\not=\mathscr{Q}_{pre}(\{f,g\})$. For the third attempt, let $\theta$ be a $1$-form on $T^*\R^n$ such that its exterior derivative is given by the symplectic form $\Omega$ on $T^*\R^n$, i.e. $\dd\theta=\Omega$. Define a covariant derivative (connection) along $X\in \mathrm{Vect}(T^*\R^n)$ by $$\nabla_X^\theta:=X-\frac{\I}{\hbar}\theta(X).$$
The idea is then to use $\nabla_{X_f}^\theta$ instead of $X_f$. 
\begin{lem}
We have 
\begin{enumerate}
\item{$[\nabla_X^\theta,f]=X(f)$,}
\item{$[X,f]=X(f)$,}
\item{$[\nabla_X^\theta,\nabla_Y^\theta]=\nabla^\theta_{[X,Y]}-\frac{\I}{\hbar}\Omega(X,Y)$, where $\Omega$ is the standard symplectic form on $T^*\R^n$.
}
\end{enumerate}
\end{lem}

\begin{proof}
We leave $(1)$ and $(2)$ as an exercise. For $(3)$, note that 
\begin{align*}
[\nabla^\theta_X,\nabla^\theta_Y]&=\left[X-\frac{\I}{\hbar}\theta(X),Y-\frac{\I}{\hbar}\theta(Y)\right]=[X,Y]-\frac{\I}{\hbar}[X,\theta(Y)]+\frac{\I}{\hbar}[Y,\theta(X)]\\
&=[X,Y]-\frac{\I}{\hbar}(X\theta(X)-Y\theta(X))\\
&=[X,Y]-\frac{\I}{\hbar}(X\theta(Y)-Y\theta(X)-\theta([X,Y])+\theta([X,Y]))\\
&=[X,Y]-\frac{\I}{\hbar}\theta([X,Y])-\frac{\I}{\hbar}(X\theta(Y)-Y\theta(X)-\theta([X,Y]))\\
&=\nabla^\theta_{[X,Y]}-\frac{\I}{\hbar}\dd\theta(X,Y)\\
&=\nabla^\theta_{[X,Y]}-\frac{\I}{\hbar}\Omega(X,Y).
\end{align*}
\end{proof}

We define the prequantum map to be given by  
\begin{equation}
\label{prequantum2}
\mathscr{Q}_{pre}(f):=-\I\hbar\nabla_{X_f}^\theta+f. 
\end{equation}
Then $\mathscr{Q}_{pre}(1)=1$. Moreover, we get the following proposition:

\begin{prop}
Let $\mathscr{Q}_{pre}$ be defined as in \eqref{prequantum2}. Then for all $f,g\in C^\infty(T^*\R^n)$ we have
\[
\frac{\I}{\hbar}[\mathscr{Q}_{pre}(f),\mathscr{Q}_{pre}(g)]=\mathscr{Q}_{pre}(\{f,g\}).
\]
\end{prop}

\begin{proof}
Indeed, we have 
\begin{align*}
\frac{\I}{\hbar}[\mathscr{Q}_{pre}(f),\mathscr{Q}_{pre}(g)]&=\frac{\I}{\hbar}[-\I\hbar\nabla^\theta_{X_f}+f,-\I\hbar \nabla^\theta_{X_g}+g]\\
&=\frac{\I}{\hbar}\Big((-\I\hbar)^2[\nabla^\theta_{X_f},\nabla^\theta_{X_g}]-\I\hbar[\nabla^\theta_{X_f},g]+\I\hbar[\nabla^\theta_{X_g},f]\Big)\\
&=\frac{\I}{\hbar}\Big((-\I\hbar)^2\nabla^\theta_{[X_f,X_g]}-(-\I\hbar)^2\frac{\I}{\hbar}\Omega(X_f,X_g)-2\I\hbar\{f,g\}\Big)\\
&=\frac{\I}{\hbar}\Big((-\I\hbar)^2\nabla^\theta_{[X_f,X_g]}-\I\hbar\{f,g\}\Big)=\mathscr{Q}_{pre}(\{f,g\})
\end{align*}
\end{proof}

\subsubsection{Prequantization on a symplectic manifold}
The goal is to generalize the constructions before to any symplectic manifold $(M,\Omega)$. We need to generalize $C^\infty(T^*\R^n)$ and in particular $L^2(T^*\R^n)$. Moreover, we need to generalize the covariant derivative $\nabla^\theta_X$. 

\begin{defn}[Complex line bundle]
A \emph{complex line bundle} $L\xrightarrow{\pi}M$ is a complex vector bundle of rank $1$, i.e. for all $x\in M$, we have $\dim L_x=1$. 
\end{defn}

\begin{ex}[trivial bundle]
Let $L:=M\times\mathbb{C}\xrightarrow{\pi}M$ be the trivial line bundle over $M$. Note that in this example, we define a section $s\colon M\to M\times \mathbb{C}$ by $s(x)=(x,1)$. Moreover, if $s'$ is any other section, then $s'(x)=f(x)\cdot s(x)$. 
\end{ex}

\begin{defn}[Nowhere vanishing section]
Let $L\xrightarrow{\pi}M$ be a line bundle over $M$. A section $s\colon M\to L$ is called \emph{nowhere vanishing} if $s(x)\in L_x\setminus\{0\}$ for all $x\in M$ (recall $L_x:=\pi^{-1}(\{x\})$).
\end{defn}

\begin{lem}
If $L\xrightarrow{\pi} M$, a complex line bundle over $M$, has a nowhere vanishing section, then $L$ is isomorphic to the trivial line bundle $M\times\mathbb{C}\xrightarrow{\pi}M$, i.e. there is a diffeomorphism $\Phi\colon L\to M\times\mathbb{C}$ such that the diagram 
\[
\begin{tikzcd}
L\arrow[d,"\pi"]\arrow[r,"\Phi"]&M\times \mathbb{C}\arrow[dl,"\mathrm{pr}_1"]\\
M&
\end{tikzcd}
\]
commutes. Moreover, for all $x\in M$, we have $\Phi\big|_{L_x}\colon L_x\to \{x\}\times\mathbb{C}$ is a vector space isomorphism.
\end{lem}
\begin{proof}
Exercise. \emph{Hint: Show that if $s$ is a nowhere vanishing section, then it defines a map $f_s\colon L\to \mathbb{C}$ such that $f_s\colon \pi^{-1}(\{x\})\xrightarrow{\sim} \mathbb{C}$ is an isomorphism}.
\end{proof}

\begin{defn}[Trivializable]
A line bundle is called \emph{trivializable} if it is isomorphic to $M\times\mathbb{C}$.
\end{defn}

\begin{rem}
A line bundle is trivializable if and only if it has a nowhere vanishing section.
\end{rem}

\begin{exe}
Let $L\xrightarrow{\pi}M$ be a trivializable line bundle and $s\colon M\to L$ be a nowhere vanishing section. Using $s$, construct a $C^\infty(M)$-linear map $\alpha_s\colon \Gamma(M,L)\to C^\infty(M)$, which is a bijection, i.e. $\alpha_s$ is a $C^\infty(M)$-module isomorphism.
\end{exe}

\begin{defn}[Hermitian metric]
A \emph{Hermitian metric} $h$ on a complex line line bundle $L\xrightarrow{\pi}M$ is a \emph{smooth} family $(h_x)_{x\in M}$, where each $h_x$ is a Hermitian form on $L_x$, which is positive-definite, i.e. we have maps $h_x\colon L_x\times L_x\to \mathbb{C}$, such that $h_x$ is sesquilinear, Hermitian, and positive-definite.
\end{defn}

\subsubsection{Connection on a line bundle}

\begin{defn}[Connection]
A \emph{connection} $\nabla$ on a line bundle $L\xrightarrow{\pi}M$ is a map 
\begin{align*}
\nabla\colon \Gamma(M,L)\times \Gamma(M,TM)&\to \Gamma(M,L)\\
(s,X)&\mapsto \nabla_Xs
\end{align*}
such that
\begin{enumerate}[$(i)$]
\item{for all $s\in \Gamma(M,L)$, $X\mapsto \nabla_Xs$ is $C^\infty(M)$-linear,}
\item{for all $X\in \Gamma(M,TM)$, $s\mapsto \nabla_Xs$ is $\mathbb{C}$-linear,}
\item{for all $f\in C^\infty(M)$, for all $X\in\Gamma(M,TM)$, and for all $s\in\Gamma(M,L)$
$$\nabla_X(fs)=X(f)s+f\nabla_Xs.$$
}
\end{enumerate}
\end{defn}

\begin{ex}[Trivial connection]
The \emph{trivial connection} on the trivial line bundle $L=M\times\mathbb{C}\xrightarrow{\pi}M$ is given by the map 
\begin{align*}
\nabla^{triv}\colon \Gamma(M,L)\times \Gamma(M,TM)&\to \Gamma(M,L)\\
(f,X)&\mapsto \nabla^{triv}_Xf:=X(f).
\end{align*}
Recall here that $\Gamma(M,L)\cong C^\infty(M)$.
\end{ex}

\begin{ex}
Let $L:=M\times \mathbb{C}\xrightarrow{\pi}M$ be the trivial line bundle over $M$ and $\theta\in\Omega^1(M)$. Define 
$$\nabla_X^\theta f:=X(f)-\frac{\I}{\hbar}\theta(X)f,$$
where $f\in C^\infty(M)$ and $X\in\Gamma(M,TM)$. Then we can check that $\nabla^\theta_X$ is indeed a connection.
\end{ex}

\begin{lem}
Let $\nabla$ be a connection on a line bundle $L\xrightarrow{\pi}M$. Let $s$ be a nowhere vanishing section of $L$. Then there is a $1$-form $\theta^s$ such that 
$$\nabla_X\tilde s=\nabla_X^{\theta^s}\tilde{s}=\left( X\left(\frac{\tilde{s}}{s}\right)-\frac{\I}{\hbar}\theta^s(X)\frac{\tilde{s}}{s}\right)s,$$
for all $\tilde{s}\in\Gamma(M,L)$ and $X\in \Gamma(M,TM)$.
\end{lem}

\begin{proof}
Consider the map $\Gamma(M,TM)\to C^\infty(M)$, $X\mapsto\left(-\frac{\hbar}{\I}\right) \frac{\nabla_Xs}{s}$. One can check that it indeed defines a $1$-form $\theta^s$. Moreover, 
\[
\nabla^{\theta^s}_X\tilde{s}=\nabla^{\theta^s}_X\left(\frac{\tilde{s}}{s}\cdot s\right)=X\left(\frac{\tilde{s}}{s}\right)s+\frac{\tilde{s}}{s}\nabla^{\theta^s}_Xs=\left( X\left(\frac{\tilde{s}}{s}\right)-\frac{\I}{\hbar}\theta^s(X)\frac{\tilde{s}}{s}\right)s.
\]
\end{proof}

\begin{rem}
Let $L\xrightarrow{\pi}M$ be a line bundle with a connection $\nabla$. Then, using a local trivialization $s\colon U\to L\big|_U$, we can find a $1$-form $\theta^s$ on $U$ such that $\nabla_X=\nabla^{\theta^s}_X=X-\frac{\I}{\hbar}\theta^s(X)$ on $U$.
\end{rem}

\subsubsection{Curvature of a connection}
Let $(L,\nabla)$ be a line bundle with connection over $M$. We define the \emph{curvature} $R^\nabla$ of $\nabla$ as the map:
\begin{align*} 
R^\nabla\colon \Gamma(M,TM)\times\Gamma(M,TM)\times \Gamma(M,L)&\to \Gamma(M,L)\\
(X,Y,s)&\mapsto R^\nabla(X,Y)s:=\I\left(\nabla_X\nabla_Y-\nabla_Y\nabla_X-\nabla_{[X,Y]}\right)s
\end{align*}
Unlike the connection $\nabla$, $R^\nabla(X,Y)\colon \Gamma(M,L)\to \Gamma(M,L)$ is $C^\infty(M)$-linear and hence defines a map $R^\nabla\colon \Gamma(M,TM)\times\Gamma(M,TM)\to \Gamma(M,\End(L))$, which is $C^\infty(M)$-linear and alternating. Note that to any line bundle $\End(L)$ is again a line bundle and is trivializable (indeed, the map $x\mapsto \id_x\colon L_x\to L_x$ defines a nowhere vanishing section of $\End(L)$). The bundle $\End(L)$ over $M$ is called the \emph{endomorphism bundle} of $L$. This implies that $\Gamma(M,\End(L))$ can be identified with $C^\infty(M)$. Hence, $R^\nabla\colon \Gamma(M,TM)\times \Gamma(M,TM)\to C^\infty(M)$ is bilinear and alternating and thus $R^\nabla$ can be identified with a $2$-form on $M$.

\begin{defn}[Prequantizable]
Let $(M,\Omega)$ be a symplectic manifold. We say that $(M,\Omega)$ is \emph{prequantizable} if there is a Hermitian line bundle $(L,\nabla)$ with a connection over $M$ such that $R^\nabla=\frac{1}{\hbar}\Omega$.
\end{defn}

\begin{ex}
Consider the symplectic manifold $(T^*N,\Omega_{can})$ for some manifold $N$. Moreover, consider the trivial bundle $L=T^*N\times\mathbb{C}\xrightarrow{\pi}\mathbb{C}$ with the connection $\nabla=\nabla^\alpha$, where $\alpha$ is the tautological $1$-form. Given a Hermitian line bundle $L\xrightarrow{\pi}(T^*N,\Omega_{can})$, we can talk about \emph{square-integrable} sections of $L$. Note that $\lambda:=\frac{1}{n!}\Omega^{\land n}$ defines a volume form on $T^*N$. Let $s\in \Gamma(T^*N,L)$, and consider the map $x\mapsto h(s(x),s(x))$. We get $h(s,s)\in C^\infty(T^*N)$. Moreover, define 
\begin{equation}
\label{inner_prod}
C(s,s):=\int_{T^*N}h(s,s)\lambda,
\end{equation}
and $\|s\|:=\left(\int_{T^*N}h(s,s)\lambda\right)^{1/2}$. We say that $s\in\Gamma(T^*N,L)$ is \emph{square-integrable} if $\|s\|<\infty$. 
\end{ex}

\begin{defn}[Square-integrable]
A \emph{square-integrable} section $s$ is an element of the \emph{completion} of smooth square-integrable sections of the line bundle $L$.
\end{defn}
We denote the space of square-integrable sections of $L$ by $\calH_{pre}$.

\begin{prop}
$\calH_{pre}$ is a Hilbert space.
\end{prop}

\subsubsection{Prequantization of $(M,\Omega)$}
We want to construct a prequantization for any symplectic manifold. Our data is a Hermitian line bundle with connection $(L,\nabla,h)$ such that $R^\nabla=\frac{1}{\hbar}\Omega$. The triple $(L,\nabla,h)$ is called a \emph{prequantum line bundle}. The \emph{prequantum Hilbert space} is given by $\calH_{pre}$. Given $f\in C^\infty(M,\R)$, we define $\mathscr{Q}_{pre}(f):=-\I\hbar\nabla_{X_f}+f$, where $X_f$ is the Hamiltonian vector field associated to $f$.

\begin{lem}
On $\Gamma(M,L)\cap\calH_{pre}$ we have $-\frac{\I}{\hbar}[\mathscr{Q}_{pre}(f),\mathscr{Q}_{pre}(g)]=\mathscr{Q}_{pre}(\{f,g\})$.
\end{lem}

\subsection{Problems with prequantization}
There are several problems that arise with the prequantization scheme as derived before. First, $\calH_{pre}$ is too \emph{big}. Moreover, $\mathscr{Q}_{pre}(f)$ is not positive (even if $f$ is).

\begin{ex}
Let $M=T^*\R$. Take $\theta=\frac{1}{2}(p\dd x-x\dd p)$, $H=\frac{1}{2}(p^2+x^2)$ ($H$ is called \emph{classical harmonic oscillator}), and $X_H=\frac{1}{2}(X_{p^2}+X_{x^2})=p\partial_x-x\partial_p$. Then $\theta(X_H)=\frac{1}{2}(p^2+x^2)$ and 
\[
\mathscr{Q}_{pre}(H)=-\I\hbar\left(p\partial_x-x\partial_p-\frac{\I}{2\hbar}(p^2+x^2)+\frac{1}{2}(p^2+x^2)\right)=-\I\hbar(p\partial_x-x\partial_p).
\]
Observe that $p\partial_x-x\partial_p$ is a vector field coming from a curl. Thus  for $r=x^2+p^2$ and $x=r\cos\phi$, $p=r\sin\phi$, we get $\mathscr{Q}_{pre}(H)(f(r)\ee^{\I n\phi})=-\I\hbar f(r)\I n\ee^{\I n\phi}=n\hbar f(r)\ee^{\I n\phi}$ for any integer $n$. Thus $n\hbar$ are eigenvalues of $\mathscr{Q}_{pre}(H)$ for all $n\in\mathbb{Z}$. This implies that $\mathscr{Q}_{pre}(H)$ has negative values and thus it is not a positive operator on $L^2(T^*\R)$.
\end{ex}

\subsection{Quantization I}
We fix the manifold $M:=T^*\R^n$ together with its standard symplectic form $\Omega$ and we set $\theta=\sum_{j=1}^np_j\dd x^j$. Let $J$ be the standard complex structure on $\R^{2n}$, which is positive and compatible with $\Omega$, i.e. 
$$J=\begin{pmatrix}0&I\\ -I&0\end{pmatrix},$$
such that $\Omega(\enspace,J\enspace)$ is the standard inner product on $\R^{2n}$. Note that $\R^{2n}$ together with the complex structure $J$ can be identified with $\mathbb{C}^n$, where the complex coordinates are given by $z=(z_1,...,z_n)$ with $z_j=x^j-\I p_j$ for $j=1,...,n$. Moreover, define the differential operators
\begin{align*}
\partial_{z_j}&:=\frac{1}{2}\left(\partial_{x^j}+\I\partial_{p_j}\right),\\
\partial_{\bar z_j}&:=\frac{1}{2}\left(\partial_{x^j}-\I\partial_{p_j}\right).
\end{align*}

A function $f\colon\mathbb{C}^n\to \mathbb{C}$ is \emph{holomorphic} if and only if $\partial_{\bar z_j}f=0$ for all $j=1,...,n$. Recall that we want to start with the prequantum Hilbert space and we want to \emph{throw away} extra information and construct a quantum Hilbert space. Consider $\calH_{pre}=L^2(\R^{2n})$ as the prequantum Hilbert space, and the position Hilbert space by $L^2(\R^n)$. We need a mechanism that allows us to select $f\in C^\infty(\R^{2n})$ which are independent of $p_1,...,p_n$. This motivates the following definitions. 

\begin{defn}[Position subspace]
The \emph{position subspace} is given by 
\[
V^{pos}:=\{f\in C^\infty(T^*\R^n)\mid \nabla_{\de_{p_j}}^\theta f=0,\,\, \forall j=1,...,n\}.
\]
\end{defn}

\begin{defn}[Momentum subspace]
The \emph{momentum subspace} is given by 
\[
V^{mom}:=\{f\in C^\infty(T^*\R^n)\mid \nabla_{\de_{x^j}}^\theta f=0,\,\, \forall j=1,...,n\}.
\]
\end{defn}

\begin{defn}[Holomorphic subspace]
The \emph{holomorphic subspace} is given by 
\[
V^{hol}:=\{f\in C^\infty(T^*\R^n)\mid \nabla_{\de_{\bar z_j}}^\theta f=0,\,\, \forall j=1,...,n\}.
\]
\end{defn}

\begin{lem}
The following hold:
\begin{enumerate}
\item{$V^{pos}$, $V^{mom}$, $V^{hol}$ are subspaces of $C^\infty(T^*\R^n)$.
}
\item{$\phi\in V^{pos}$ if and only if $\partial_{p_j}\phi=0$ for all $j=1,...,n$., i.e. $\phi(x,p)=\psi(x)$ for $\psi\in C^\infty(\R^n)$.
}
\item{$\phi\in V^{mom}$ if and only if $\phi(x,p)=\ee^{\frac{\I}{\hbar}x\cdot p}\psi(p)$ for $\psi\in C^\infty(\R^n)$ such that $\partial_{x^j}\psi=0$ for all $j=1,...,n$.
}
\item{$\phi\in V^{hol}$ if and only if $\phi(x,p)=\ee^{-\frac{p^2}{2\hbar}}F(z)$, where $F$ is holomorphic on $\mathbb{C}^n$.

}
\end{enumerate}
\end{lem}

\begin{proof}
$(1)$ is obvious. For $(2)$, note that $\theta(\partial_{p_j})=0$ for all $j$. Thus $\nabla_{\partial_{p_j}}^\theta=\partial_{p_j}$ and hence $\nabla_{\partial_{p_j}}^\theta\phi=0$ if and only if $\partial_{p_j}\phi=0$. For $(3)$, we note that $\theta(\partial_{x^j})=p_j$. Thus $\nabla_{\partial_{x^j}}^\theta=\partial_{x^j}-\frac{\I}{\hbar}p_j$. Now $\nabla_{\partial_{x^j}}^\theta\left(\ee^{\frac{\I}{\hbar}x\cdot p}\psi\right)=\partial_{x^j}\left(\ee^{\frac{\I}{\hbar}x\cdot p}\psi\right)-\frac{\I}{\hbar}p_j\ee^{\frac{\I}{\hbar}x\cdot p}\psi=\ee^{\frac{\I}{\hbar}x\cdot p}\partial_{x^j}\psi$. This implies that $\nabla_{\partial x^j}^\theta\left(\ee^{\frac{\I}{\hbar}x\cdot p}\psi\right)=0$ if and only if $\partial_{x^j}\psi=0$ and thus $\phi\in V^{mom}$ if and only if $\phi=\ee^{\frac{\I}{\hbar}x\cdot p}\psi$ with $\partial_{x^j}\psi=0$ for all $j$. Finally, for $(4)$, we see that $\theta(\partial_{\bar z_j})=\frac{1}{2}p_j$ and thus $\nabla_{\partial_{\bar z_j}}^\theta=\partial_{\bar z_j}-\frac{\I}{2\hbar}p_j$. This implies 
$$\nabla_{\partial_{\bar z_j}}^\theta\left(\ee^{-\frac{p^2}{2\hbar}}F\right)=\partial_{\bar z_j}\left(\ee^{-\frac{p^2}{2\hbar}}F\right)-\frac{\I}{\hbar}p_j\ee^{-\frac{p^2}{2\hbar}}F=\ee^{-\frac{p^2}{2\hbar}}\partial_{\bar z_j}F-\frac{1}{2\hbar}F\ee^{-\frac{p^2}{2\hbar}}\partial_{\bar z_j}p^2-\frac{\I}{\hbar}p_j\ee^{-\frac{p^2}{2\hbar}}F.$$
This implies that $\nabla_{\partial_{z_j}}^\theta\left(\ee^{-\frac{p^2}{2\hbar}}\right)=\ee^{-\frac{p^2}{2\hbar}}\partial_{\bar z_j}F$ and thus $\nabla_{\partial_{z_j}}^\theta\left(\ee^{-\frac{p^2}{2\hbar}}\right)=0$ if and only if $\partial_{\bar z_j}F=0$. Hence, $\phi\in V^{hol}$ if and only if $\phi=\ee^{-\frac{p^2}{2\hbar}}F(z)$, where $F$ is holomorphic.
\end{proof}
Next we want to construct Hilbert spaces using $V^{pos}$, $V^{mom}$ and $V^{hol}$. We want to start with a naive approach: Let $\phi,\psi\in V^{pos}$ and define 
\begin{equation}
\label{inner_prod}
\langle \phi,\psi\rangle_{\calH_{pos}}:=\int_{\R^{2n}}\bar\phi\psi\dd p_1\dotsm\dd p_n\dd x^{1}\dotsm\dd x^n.
\end{equation}

Moreover, define $\calH_{pos}$ as the completion of $\{\phi\in V^{pos}\mid \|\phi\|^2<\infty\}$, where $\|\enspace\|$ is given by \eqref{inner_prod}. The problem in this approach is that $L^2(\R^{2n})\cap V^{pos}=\{0\}$ and hence $\calH_{pos}=\{0\}$. Using the naive approach, we can not construct a nontrivial Hilbert space out of $V^{pos}$. The same argument shows that we can not get a nontrivial Hilbert space $V^{mom}$. However, next we show that the naive approach will lead to a Hilbert space $\calH_{hol}$ from $V^{hol}$, which is usally called the \emph{Segal-Bergmann space} used in many Quantum Mechanics text books as a quantum Hilbert space. Let $\phi,\psi\in V^{hol}$ with $\phi=\ee^{-\frac{p^2}{2\hbar}}F$ and $\psi=\ee^{-\frac{p^2}{2\hbar}}G$. Define then $$\langle\phi,\psi\rangle_{\calH_{hol}}:=\int_{\R^{2n}}\bar\phi\bar\psi\dd p_1\dotsm\dd p_n\dd x^{1}\dotsm\dd x^n=\int_{\R^{2n}}\bar F G\ee^{-\frac{p^2}{2\hbar}}\dd p_1\dotsm\dd p_n\dd x^{1}\dotsm\dd x^n.$$
Then se set $\calH_{hol}$ to be the completion of $\{\phi\in V^{hol}\mid \| \phi\|^2<\infty\}$. In contrast to $\calH_{pos}$, we will show that $\calH_{hol}$ is an infinite-dimensional Hilbert space.

\begin{lem}
The following hold:
\begin{enumerate}
\item{Let $\psi_k(x,p)=z^k\ee^{-\frac{z^2}{4\hbar}}\ee^{-\frac{p^2}{2\hbar}}$ for $k\in\N$. Then $\psi_k\in\calH_{hol}$. Note that here $z^2:=\sum_{j=1}^nz_j^2$.
}
\item{$\calH_{hol}\cong \mathscr{H}L^2(\mathbb{C}^n,\nu)$, where $\mathscr{H}L^2$ denotes the holomorphic $L^2$-space and $$\dd\nu=\ee^{-\frac{p^2}{2\hbar}}\dd p_1\dotsm\dd p_n\dd x^{1}\dotsm\dd x^n.$$
}
\end{enumerate}
\end{lem}

\begin{proof}
We start with $(1)$. It is easy to see that $\psi_k\in V^{hol}$. We will show that $\psi_k\in\ L^2(\R^{2n})$. Note that $z^2+\bar z^2=2\sum_{j=1}^n(x_j^2-p_j^2)=2(x^2-p^2)$. Thus, we have $$\ee^{-\frac{z^2}{4\hbar}-\frac{\bar z^2}{4\hbar}}=\ee^{-\frac{x^2}{2\hbar}}\ee^{-\frac{p^2}{2\hbar}}.$$ Hence, we get
$$\int_{\R^{2n}}\vert \psi_k\vert^2\dd p_1\dotsm\dd p_n\dd x^{1}\dotsm\dd x^n=\int_{\R^{2n}}\vert z\vert^{2k}\ee^{-\frac{x^2}{2\hbar}}\ee^{-\frac{p^2}{2\hbar}}\dd p_1\dotsm\dd p_n\dd x^{1}\dotsm\dd x^n<\infty.$$
Note that $\ee^{-\frac{(x^2+p^2)}{2\hbar}}$ gives a Gaussian measure and since $\vert z\vert^{2k}$ is polynomial, we get finiteness. This shows tat $\psi_k\in L^2(\R^{2n})$ for all $k\in\N$. For $(2)$, note that $\phi\in V^{hol}$ if and only if $\phi=\ee^{-\frac{p^2}{2\hbar}}F$, where $F$ is holomorphic. Thus we have a map $L\colon\calH_{hol}\to \mathscr{H}(\mathbb{C}^n)$ which is given by $L(\phi)=\phi\ee^{\frac{p^2}{2\hbar}}$. We have denoted by $\mathscr{H}(\mathbb C^n)$ the space of holomorphic functions $\mathbb C^n\to \mathbb C$. Moreover, 
\begin{multline*}
\int_{\R^{2n}}\vert\phi\vert^2\dd p_1\dotsm\dd p_n\dd x^{1}\dotsm\dd x^n=\int_{\R^{2n}}\vert F\vert^2\ee^{-\frac{p^2}{\hbar}}\dd p_1\dotsm\dd p_n\dd x^{1}\dotsm\dd x^n\\=\int_{\R^{2n}}\vert L(\phi)\vert^2\underbrace{\ee^{-\frac{p^2}{\hbar}}\dd p_1\dotsm\dd p_n\dd x^{1}\dotsm\dd x^n}_{\dd\nu}.
\end{multline*}
Hence $L(\phi)\in L^2(\R^{2n},\nu)$, which implies that 
\[
L(\phi)\in L^2(\R^{2n},\nu)\cap \mathscr{H}(\mathbb{C}^n)=:\mathscr{H}L^2(\mathbb{C}^n,\nu).
\]
Moreover, $L^{-1}\colon \mathscr{H}L^2(\mathbb{C}^n,\nu)\to \calH_{hol}$ is given by $L^{-1}(F)=\ee^{-\frac{p^2}{2\hbar}}F$. Note that $L$ is an isomorphism of Hilbert spaces. One can show that $z^k\ee^{-\frac{z^2}{4\hbar}}$ forms an orthogonal basis of $\mathscr{H}L^2(\mathbb{C}^n,\nu)$, which is thus infinite-dimensional and hence $\calH_{hol}$ is infinite-dimensional as well. 
\end{proof}

\begin{rem}
\label{rem_theta}
Instead of taking $\theta=\sum_{j=1}p_j\dd x^j$, we can take $\widetilde{\theta}=\frac{1}{2}\sum_{j=1}^n(p_j\dd x^j-x^j\dd p_j)$ and the connection $\nabla^{\widetilde{\theta}}$. In this case we have $\phi\in V^{hol}$ if and only if $\phi\in\ee^{-\frac{z^2}{4\hbar}}F$, where $F$ is holomorphic. One can define $\calH_{hol}$ as before and one can show that $\calH_{hol}$ is isomorphic to (as Hilbert spaces) $\mathscr{H}L^2(\mathbb{C}^n,\mu)$, where 
\[
\dd\mu=\ee^{-\frac{z^2}{2\hbar}}\dd p_1\dotsm \dd p_n\dd x^1\dotsm \dd x^n.
\]
\end{rem}

\begin{ex}
We want to look at the case $n=1$. Let $M:=T^*\R\ni (x,p)$ and consider $\theta=\frac{1}{2}(p\dd x-x\dd p)$. From Remark \ref{rem_theta} we know $\calH_{hol}\cong \mathscr{H}L^2(\mathbb{C},\mu)$. Moreover, one can check that $\psi_k=z^k\ee^{-\frac{z^2}{4\hbar}}$ gives an orthogonal basis of $\calH_{hol}$. Furthermore, we can show that $$\mathscr{Q}_{pre}(H)\left(\ee^{-\frac{z^2}{4\hbar}}F\right)=\hbar z\ee^{-\frac{z^2}{4\hbar}}\frac{\dd F}{\dd z},$$
where $H(x,p):=\frac{1}{2}(x^2+p^2)$ is the harmonic oscillator. Thus $\mathscr{Q}(H)(\psi_k)=k\hbar\psi_k(z)$ for all $k\in\N$ and hence $\psi_k(z)$ are eigenvectors associated to the eigenvalues $k\hbar$ for $k\in\N$. Since all eigenvalues are nonnegative and $\psi_k$ forms a basis of $\calH_{hol}$, we get that $\mathscr{Q}_{pre}(H)$ is a nonnegative operator on $\calH_{hol}$. This example shows that we are able to improve one of the drawbacks of prequantization.
\end{ex}

\begin{rem}
$k\hbar$, for $k\in\N$, are not true answers for eigenvalues of the harmonic oscillator ($\frac{1}{2}\hbar$ is missing). This can be achieved by using \emph{half-form} quantization.
\end{rem}

\subsection{Quantization II}
We have seen that the naive approach to quantization may or may not lead to a construction of a reasonable Hilbert space. Next, our goal will be to outline a construction called \emph{half-form} quantization, which might lead to \emph{correct} Hilbert spaces. At least, we will see that we can construct position and momentum Hilbert spaces. 

\begin{defn}[Distribution]
Let $M$ be a smooth manifold. A real (complex) \emph{distribution} of rank $k$, where $k\leq \dim M$, is a subbundle $D$ of $TM$ ($TM^\mathbb{C}$) such that $D_x\subseteq T_xM$ for al $x\in M$ with $\dim_\R(D_x)=k$ for the real case, and $D_x\subseteq T_xM^\mathbb{C}$ with $\dim_\mathbb{C}D_x=k$ for the complex case.
\end{defn}

\begin{rem}
Let $D$ be a distribution on $M$. We will use $\Gamma(M,D)$ to denote the space of sections of $D$. Given a distribution $D$ on $M$, we can talk about functions on $M$, which are constant in the \emph{direction} of $D$. More precisely, we say $f\in C^\infty(M)$ is constant along $D$ if $X(f)=0$ for all $X\in \Gamma(M,D)$. We will use $C^\infty_D(M)$ to denote functions on $M$, which are constant along $D$.
\end{rem}

\begin{ex}
Let $M:=\R^2\ni(x,p)$. Let $D(x,p):=\mathrm{span}_\R\{\partial_x\}$. Then $D$ is a distribution and $\Gamma(M,D)=\mathrm{span}_{C^\infty(\R^2,\R)}\{\partial_x\}$.
\end{ex}

\begin{ex}
Let $M:=\R^2$ and $D_{(x,p)}:=\mathrm{span}_\R\{\partial_p\}$. Then $D$ is a real distribution and $\Gamma(M,D)=\mathrm{span}_{C^\infty(\R^2,\R)}\{\partial_p\}$. $D$ is called a \emph{vertical distribution} on $M\cong T^*\R$. 
\end{ex}

\begin{ex}
More generally, take $M:=T^*Q\xrightarrow{\pi}Q$. Then we can define a real distribution $D$ by $D_m=\ker(\dd\pi_m\colon T_mM\to T_{\pi(m)}Q)$ for each $m\in M$. Let $\dim Q=n$. Then we see than $\dim (\ker\dd_m\pi)=n$ for all $m\in M$. Let $(x^1,...,x^n,p_1,...,p_n)$ be local coordinates in a neighborhood of $m\in M$. Then we can check that $$\ker(\dd_m\pi)=\mathrm{span}_\R\{\partial_{p_1},...,\partial_{p_n}\}.$$ This distribution $D$ is called \emph{vertical distribution} on $T^*Q$.
\end{ex}

\begin{ex}
If we complexify a real distribution, we get a complex distribution. Let $M:=T^*\R^n$ and consider the standard complex structure $J$ on $\R^n$ together with $D(x,p):=\mathrm{span}_\mathbb{C}\{\partial_{z_1},...,\partial_{z_n}\}$. Hence $D$ is a complex distribution. 
\end{ex}

Given a distribution $D$ on $M$ and a complex line bundle $L$ with a connection $\nabla$ on $L$, we can talk about \emph{covariantly constant} sections of $L$ along $D$ as follows:

\begin{defn}[Covariantly constant]
A section $s\colon M\to L$ is \emph{covariantly constant} along $D$ if $\nabla_Xs=0$ for all $X\in \Gamma(M,D)$.
\end{defn}

\begin{rem}
We write $\Gamma_D(M,L):=\{s\in\Gamma(M,L)\mid \nabla_Xs=0,\forall X\in\Gamma(M,D)\}$.
\end{rem}

\begin{ex}
\label{ex:example1}
Let $M:=T^*\R^n\ni (x^1,...,x^n,p_1,...,p_n)$ and let $D$ be the vertical distribution. Let $L:=M\times\mathbb{C}\xrightarrow{\pi}\mathbb{C}$ and consider the 1-form $\theta=\sum_{j=1}^np_j\dd x^j$. then $\nabla^\theta$ is a connection on $L$. Using the identification $\Gamma(M,L)=C^\infty(M)$, we see that covariantly constant sections of $L$ are essentially the functions $f$ satisfying $\partial_{p_j}f=0$ for all $j=1,...,n$. Hence, we get $\Gamma_D(M,L)= V^{pos}$.
\end{ex}

\begin{ex}
Let $M$, $\theta$, $\nabla^\theta$ and $L$ be as in Example \ref{ex:example1}. Define
\begin{align*}
D_{(x,p)}&:=\mathrm{span}_\mathbb{C}\{\partial_{z_1},...,\partial_{z_n}\},\\
\widetilde{D}_{(x,p)}&:=\mathrm{span}_\mathbb{C}\{\partial_{\bar z_1},...,\partial_{\bar z_n}\}.
\end{align*}
Then $D$ and $\widetilde{D}$ are complex distributions and $\Gamma_{\widetilde{D}}(M,L)= V^{hol}$.
\end{ex}

We have seen that given a distribution $D$ on a manifold $M$ and a line bundle $(L,\nabla)$ with a connection over $M$, we can talk about sections of $L$, which are covariantly constant along $D$. In principle, it can happen that $\Gamma_D(M,L)$ is $\{0\}$ or too small. We want to understand what properties $D$ should have such that $\Gamma_D(M,L)$ is \emph{as big as possible}. Let $\hbar=1$ form now on. Let $(M,\Omega)$ be a symplectic manifold and $(L,\nabla)$ be a prequantum line bundle on $M$, i.e. $R^\nabla(X,Y)=\Omega(X,Y)$ for all $X,Y\in\Gamma(M,TM)$. Let $D$ be a distribution on $M$ and $\phi\in\Gamma_D(M,L)$. Then for all $X,Y\in\Gamma(M,D)$ we have $\nabla_X\phi=0$ and $\nabla_Y\phi=0$. Hence $[\nabla_X,\nabla_Y]\phi=0$. Recall that $[\nabla_X,\nabla_Y]=\nabla_{[X,Y]}-\I\Omega(X,Y)$ and thus 
\begin{equation}
\label{eq1}
\nabla_{[X,Y]}\phi-\I\Omega(X,Y)\phi=0.
\end{equation}
If we assume $[X,Y]\in\Gamma(M,D)$, then \eqref{eq1} implies that $\Omega(X,Y)=0$. From \eqref{eq1} and $\Omega(X,Y)=0$, we can see that if the distribution $D$ satisfies \eqref{eq1} and $\Omega(X,Y)=0$ for $X,Y\in\Gamma(M,D)$. Then the necessary condition $[\nabla_X,\nabla_Y]\phi=0$ holds and hence there is a chance that we get a reasonably \emph{big} $\Gamma_D(M,L)$. This motivates the following definition.

\begin{defn}[Real polarization]
Let $(M,\Omega)$ be a symplectic manifold. A \emph{real polarization} of $M$ is a real distribution $D$ such that 
\begin{enumerate}
\item{$X,Y\in\Gamma(M,D)$ implies that $[X,Y]\in\Gamma(M,D)$. This condition means that $D$ is \emph{involutive} (or \emph{integrable}).
}
\item{$D_x$ is a \emph{Lagrangian subspace} of $T_xM$ for all $x\in M$, i.e. for all $u,v\in D_x$, $\Omega(u,v)=0$ and $\dim D_x=\frac{1}{2}\dim M$.}
\end{enumerate}
\end{defn}

\begin{defn}[Complex polarization]
Let $(M,\Omega)$ be a symplectic manifold. A \emph{complex polarization} of $M$ is a complex distribution $D$ such that 
\begin{enumerate}
\item{for all $X,Y\in\Gamma(M,D)$ we get $[X,Y]\in\Gamma(M,D)$ ($D$ is \emph{integrable}).
}
\item{$D_x$ is a \emph{Lagrangian subspace} of $T_xM^\mathbb{C}$ for all $x\in M$.}
\item{$\dim(D_x\cap\overline{D_x})$ is constant in $x\in M$.}
\end{enumerate}
\end{defn}

\begin{rem}
We can observe that for a real polarization $D$ of $M$, the complexification $D^\mathbb{C}$ of $D$ is a complex polarization because $D_x\cap\overline{D_x}=D_x$ for all $x\in M$.
\end{rem}

\begin{ex}
Let $M:=T^*Q$ and $P$ the vertical distribution. Then $P$ is a polarization. $P$ is called the \emph{vertical polarization}.
\end{ex}

\begin{ex}
Let $M:=T^*\R^{n}$ and $J$ be the standard complex structure as before. Moreover, consider
\begin{align*}
   P_{(x,p)}&:=\mathrm{span}_\mathbb{C}\{\partial_{z_1},...,\partial_{z_n}\},\\
   \overline{P}_{(x,p)}&:=\mathrm{span}_\mathbb{C}\{\partial_{z_1},...,\partial_{z_n}\}.
\end{align*}
Then $P$ and $\overline{P}$ are complex polarizations.
\end{ex}

\begin{defn}[Involutive distribution]
A (real) distribution $D$ on $M$ with the property that $X,Y\in \Gamma(M,D)$ implies $[X,Y]\in\Gamma(M,D)$ is called \emph{involutive}.
\end{defn}

\begin{rem}
If a real distribution $D$ is involutive, there is a \emph{foliation} of $M$ by integral submanifolds of $D$, i.e. there exists a collection $\{S_i\}_{i\in I}$ of submanifolds of $M$ such that all the $S_i$ are mutually disjoint and $M=\bigsqcup_{i\in I}S_i$ (this is the foliation part). Moreover, for all $x\in S_i$, we have $T_xS_i=D_x$ (this is the integral submanifold part). Each $S_i$ is called a \emph{leaf} on the foliation (equally, \emph{leaf} of $D$). Given an involutive distribution $D$ and an associated foliation $\{S_i\}_{i\in I}$, we can define a new topological space, which is the space of equivalence classes of $M$, where the equivalence relation $\sim$ arises from the foliation: for $x,y\in M$ we have $x\sim y$ if and only if there is an $i\in I$ such that $x,y\in S_i$.
\end{rem}

Assume that $M/D$ is a smooth manifold and $\pi\colon M\to M/D$, which is the canonical projection, is smooth. 

\begin{ex}
Let $M:=T^*N$, for some manifold $N$, and $D$ be the vertical distribution on $M$. Then a leaf is exactly a fiber of $T^*N$ over $N$. In this case $M/D$ is diffeomorphic to $N$ and $\pi$ can be identified with the usual projection map $T^*N\to N$.
\end{ex}

\begin{ex}
Let $M:=\R^2\ni(x,p)$ and let $D$ be the horizontal distribution, i.e. $$\Gamma(M,D)=\mathrm{span}_{C^\infty(M,\R)}\{\partial_x\}.$$ Then $M/D\cong \R$ and $\pi\colon M\to M/D$ can be identified with the projection $(x,p)\mapsto p$.
\end{ex}

\begin{ex}
Let $M:=\R^2\setminus\{(0,0)\}$ and let $D$ be the distribution for which $$\Gamma(M,D)=\mathrm{span}_{C^\infty(M,\R)}\{x\partial_p-p\partial_x\}.$$ Then $M/D\cong \R^+$ and $\pi\colon M\to M/D$ can be identified with the map $(x,p)\mapsto x^2+p^2$.
\end{ex}

\subsubsection{Half-form quantization (real case)}
Let $(M,\Omega)$ be a symplectic manifold and $P$ a real polarization. Assume that the space of leaves $N:=M/P$ is a smooth manifold and $\pi\colon M\to N$ is smooth. Moreover, define a line bundle $K_P$ as follows: An $n$-form $\alpha$ is a section of $K_P$ if and only if $\iota_X\alpha=0$ for all $X\in\Gamma(M,P)$. 

\begin{defn}[Canonical bundle]
The bundle $K_P$ is called the \emph{canonical bundle} of $P$ and is given by 
$$(K_P)_x:=\bigwedge^n\left(\mathrm{Ann}(P_x)\right),$$
where $\mathrm{Ann}(P_x):=\{\alpha\in T^*_xM\mid \alpha(u)=0,\forall u\in P_x\}$ denotes the annihilator of $P_x$.
\end{defn}

\begin{defn}[$P$-polarized form]
We say that $\alpha\in\Gamma(M,K_P)$ is \emph{$P$-polarized} if $\iota_X(\dd\alpha)=0$ for all $X\in\Gamma(M,P)$.
\end{defn}

\begin{ex}
Let $M:=T^*\R^n\ni (x^1,...,x^n,p_1,...,p_n)$ and $P$ be the vertical polarization on $M$. Then $N\cong \R^n\ni(x^1,...,x^n)$ and $M\to N$ the projection. Consider an $n$-form $\alpha=f(x,p)\dd x^1\land\dotsm \land\dd x^n\land \dd p_1\land\dotsm \land \dd p_n$ on $M$. Then $\iota_{\partial_{p_j}}\alpha=0$ for all $j=1,...,n$ if and only if $\alpha=f(x,p)\dd x^1\land\dotsm  \land \dd x^n$. Moreover, 
\begin{equation}
\label{moreover}
\iota_{\partial_{p_j}}(\dd\alpha)=0\Longleftrightarrow\partial_{p_j}f=0, \hspace{0.3cm}\forall j=1,...,n.
\end{equation}
\end{ex}

\begin{exe}
Check \eqref{moreover}.
\end{exe}

\begin{rem}
A $P$-polarized section of $K_P$ has the form $f(x)\dd x^1\land\dotsm\land \dd x^n$.
\end{rem}

\begin{prop}
Let $\tilde{\alpha}$ be an $n$-form on $N$. Then $\pi^*(\tilde{\alpha})$ is a $P$-polarized section of $K_P$. If $\alpha$ is a $P$-polarized section of $K_p$, then $\alpha=\pi^*(\tilde{\alpha})$ for some $n$-form on $N$.
\end{prop}

\begin{defn}[square root of a line bundle]
Let $L$ be a line bundle on a manifold $M$. A line bundle $Q\to M$ is a \emph{square root} of $L$ if there exists an isomorphism $Q\otimes Q\to L$. 
\end{defn}

\begin{ex}[Trivial line bundle]
A \emph{trivial line bundle} has a square root: There is an isomorphism between $M\times (\mathbb{C}\otimes \mathbb{C})\to M$ and $\mathbb{C}\to M$ coming from the isomorphism $\mathbb{C}\otimes_\mathbb{C}\mathbb{C}\xrightarrow{\sim} \mathbb{C}$.
\end{ex}

\begin{rem}
We want to assume that $K_P$ has a square root and we fix a square root $S_P$ of $K_P$ from now on.
\end{rem}

\begin{ex}
Let $M:=T^*\R^n$ and let $P$ be the vertical polarization on $M$. We have seen that $\Gamma(M,K_P)=\{f\dd x^1\land\dotsm \land \dd x^n\mid f\in C^\infty(T^*\R^n)\}$. Hence 
$$\Gamma(M,S_P)=\left\{f\sqrt{\dd x^1\land\dotsm \land\dd x^n}\,\,\big|\,\, f\in C^\infty(T^*\R^n)\right\},$$
where $\sqrt{\dd x^1\land \dotsm \land \dd x^n}$ is just a notation to indicate the fact that 
$$\sqrt{\dd x^1\land \dotsm \land \dd x^n}\otimes\sqrt{\dd x^1\land \dotsm \land \dd x^n}=\dd x^1\land\dotsm \land\dd x^n.$$
\end{ex}

\begin{rem}
\label{rem_form}
Recall that, given $\alpha\in\Gamma(M,K_P)$ and $X\in\Gamma(M,P)$, we get $\iota_X(\dd\alpha)\in\Gamma(M,K_P)$.
\end{rem}

Using Remark \ref{rem_form}, we can define a \emph{partial} connection $\nabla_X^{part,P}\alpha:=\iota_X(\dd\alpha)$ for a fixed $X\in\Gamma(M,P)$ (this is partial because we can not define this for all $X\in\Gamma(M,TM)$). In fact, $\nabla^{part,P}$ induces a partial connection on $S_P$, which is roughly given by solving 
$$\underbrace{\nabla_X^{part,P}(s_1\otimes s_2)}_{\in\Gamma(M,K_P)}=\underbrace{\nabla_X^{part,P}s_1}_{\in\Gamma(M,S_P)}\otimes \underbrace{s_2}_{\in\Gamma(M,S_P)}+s_1\otimes\underbrace{\nabla_X^{part,P} s_2}_{\in\Gamma(M,S_P)}.$$

\begin{rem}
We say that a section $\mu\in\Gamma(M,S_P)$ is $P$-polarized, if 
\[
\nabla_X^{part,P}\mu=0,\quad \forall X\in \Gamma(M,P).
\]
\end{rem}

\begin{ex}
Let $M:=T^*\R^n$ and $P$ the vertical polarization on $M$. Then $$\Gamma(M,S_P)=\left\{f(x,p)\sqrt{\dd x^1\land\dotsm\land \dd x^n}\,\,\big|\,\, f\in C^\infty(T^*\R^n)\right\}.$$ We get that $f(x,p)\sqrt{\dd x^1\land\dotsm \land\dd x^n}$ is $P$-polarized if and only $\partial_{p_j}f=0$ for all $j=1,...,n$.
\end{ex}

\subsubsection{Construction of the Hilbert space}
For the construction of a Hilbert space, we need to start with the following data:
\begin{itemize}
\item{A symplectic manifold $(M,\Omega)$,}
\item{A prequantum line bundle $(L,\nabla)$ on $M$ with metric $h^F$,}
\item{A real polarization $P$ on $M$,}
\item{A square root $S_P$ of the canonical line bundle associated to $P$.}
\end{itemize}
We assume that $N:=M/P$ is a smooth manifold and the projection $\pi\colon M\to N$ is smooth, that $N$ is oriented and $\nabla$ is compatible with $h^F$. Let $s_1$ and $s_2$ be $P$-polarized sections of $L$. Then 
$$X\left(h^F(s_1,s_2)\right)=h^F(\nabla_Xs_1,s_2)+h^F(s_1,\nabla_Xs_2)=0,$$
and thus $h^F(s_1,s_2)$ is a function on $N$. Let us consider the space 
$$\Gamma_P(L\otimes S_P^\mathbb{C}):=\{\text{$P$-polarized sections of $L\otimes S_P^\mathbb{C}$}\}.$$
Note that $\Gamma_P(L\otimes S_P^\mathbb{C})$ is generated by elements of the form $s\otimes\mu$, where $s$ is a $P$-polarized section of $L$ and $\mu$ is a $P$-polarized section of $S_P^\mathbb{C}$. We define 
\[
\langle s_1\otimes\mu_1,s_1\otimes \mu_2\rangle_{\calH_P}:=\int_Nh^F(s_1,s_2)\bar\mu_1\otimes\mu_2,
\]
and then extend sesquilinearly. Note that $\bar\mu_1\otimes\mu_2$ is an $n$-form on $N$. Consider the inner product space, which consists of $\gamma\in\Gamma_P(L\otimes S_P^\mathbb{C})$ for which $\|\gamma\|_{\calH_P}<\infty$. The half-form Hilbert space is then given by the completion of this inner product space with respect to this norm.

\begin{ex}
Let $M:=T^*\R^n$ with its canonical symplectic form $\Omega$. Let $P$ be the vertical polarization on $M$, $L:=M\times\mathbb{C}$ the trivial line bundle and $\theta:=\sum_{j=1}^np_j\dd x^j$ together with the induced connection $\nabla^\theta$. Then 
$$\Gamma_P(M,S_P)=\left\{ f(x)\sqrt{\dd x^1\land\dotsm \land \dd x^n}\,\,\big|\,\,  f\in C^\infty(\R^n)\right\}.$$
Since 
$$\left\langle s_1\otimes f_1\sqrt{\dd x^1\land\dotsm \land \dd x^n},s_1\otimes f_2\sqrt{\dd x^1\land\dotsm \land \dd x^n}\right\rangle_{\calH_P}=\int_{\R^n}\bar s_1s_2\bar f_1 f_2\dd x^1\land\dotsm \land \dd x^n,$$
we can identify $\calH_P$ with $L^2(\R^n)$.
\end{ex}

\printbibliography
\end{document}